\documentclass[10pt]{article}
\usepackage[english]{babel}
\usepackage[utf8]{inputenc}
\usepackage[T1]{fontenc}
\usepackage[a4paper, margin=1.00in]{geometry}
\usepackage{graphicx}
\usepackage{framed}
\usepackage[framemethod=tikz]{mdframed}
\usepackage{todonotes}
\usepackage{color}

\definecolor{darkgreen}{rgb}{0,0.5,0}
\usepackage{hyperref}
\hypersetup{
    unicode=false,          
    colorlinks=true,        
    linkcolor=blue,          
    citecolor=darkgreen,        
    filecolor=magenta,      
    urlcolor=cyan           
}

\usepackage{amsthm}
\usepackage{amsmath}
\usepackage{amssymb}
\usepackage{amsfonts}
\usepackage{mathrsfs}
\usepackage{mathtools}
\usepackage{verbatim}
\usepackage{footnote}
\usepackage{algpseudocode}
\usepackage[linesnumbered,boxed, vlined]{algorithm2e}
	\DontPrintSemicolon
	\SetKwRepeat{Do}{do}{while}
  \SetKwInOut{Input}{Input}
  \SetKwInOut{Output}{Output}
\usepackage{lineno}
\usepackage{caption}
\usepackage{framed}
\usepackage{enumerate}
\usepackage{enumitem}

\usepackage{tikzsymbols}
\usepackage{thmtools,thm-restate}
\usepackage{nicefrac}

\usepackage{subcaption}

\newtheorem{theorem}{Theorem}[section]
\newtheorem{lemma}[theorem]{Lemma}
\newtheorem{corollary}[theorem]{Corollary}
\newtheorem{definition}[theorem]{Definition}

\newtheorem*{remark*}{Remark}

\usepackage[capitalize, nameinlink]{cleveref}
\usepackage[section]{placeins}
\crefname{theorem}{Theorem}{Theorems}
\Crefname{lemma}{Lemma}{Lemmas}
\Crefname{invariant}{Invariant}{Invariants}
\Crefname{claim}{Claim}{Claims}
\Crefname{observation}{Observation}{Observations}
\Crefname{algorithm}{Algorithm}{Algorithms}
\Crefname{figure}{Figure}{Figures}

\DeclareMathOperator{\poly}{poly}
\DeclareMathOperator{\idx}{idx}
\DeclareMathOperator{\gain}{gain}
\newcommand{\E}[1]{{\mathbb{E}}\left[#1\right]}
\newcommand{\prob}[1]{\Pr \left[ #1 \right]}
\newcommand{\rb}[1]{\left(#1\right)}
\newcommand{\cA}{\mathcal{A}}
\newcommand{\cB}{\mathcal{B}}
\newcommand{\cC}{\mathcal{C}}
\newcommand{\cI}{\mathcal{I}}
\newcommand{\cP}{\mathcal{P}}
\newcommand{\cR}{\mathcal{R}}
\newcommand{\cL}{\mathcal{L}}
\newcommand{\cW}{\mathcal{W}}
\newcommand{\cY}{\mathcal{Y}}

\newcommand{\bbN}{\mathbb{N}}

\newcommand{\tG}{\tilde{G}}
\newcommand{\tM}{\tilde{M}}
\newcommand{\tO}{\tilde{O}}
\newcommand{\eps}{\epsilon}
\newcommand{\Mdiff}{M_{\rm diff}}
\newcommand{\Mblue}{M_{\rm blue}}
\newcommand{\Mred}{M_{\rm red}}
\newcommand{\Mopt}{M^{\star}}

\newcommand{\arc}[1]{\vec{#1}}
\newcommand{\jstar}{j^{\star}}
\newcommand{\decompress}{\textsc{Decompress}\xspace}
\newcommand{\compress}{\textsc{Compress}\xspace}
\newcommand{\AlgLayered}{\textsc{Alg-Layered}\xspace}
\newcommand{\AlgAlternating}{\textsc{Alg-Alternating}\xspace}
\newcommand{\AlgResolve}{\textsc{Alg-Resolve}\xspace}
\newcommand{\AlgExtracintAlternations}{\textsc{Alg-Extracting-Alternations}\xspace}
\newcommand{\AlgResolveWithinLayered}{\textsc{Alg-Resolve-within-Layered}\xspace}
\newcommand{\AlgResolveBetweenLayered}{\textsc{Alg-Resolve-between-Layered}\xspace}

\newcommand{\rresolve}{r_{\rm{resolve}}}

\title{Massively Parallel Algorithms for $b$-Matching\footnote{
Mohsen Ghaffari and Christoph Grunau received funding from the European Research Council (ERC) under the European Unions Horizon
2020 research and innovation programme (grant agreement No. 853109).
Slobodan Mitrovi\'c is grateful to Prof.~Dragan Ma\v sulovi\' c for hosting him at the University of Novi Sad while a part of this project was undertaken.
}}
\author{Mohsen Ghaffari \\ \small{ETH Zurich} \\ \small{ghaffari@inf.ethz.ch} 
\and
Christoph Grunau \\ \small{ETH Zurich} \\ \small{christoph.grunau@inf.ethz.ch}
\and
Slobodan Mitrovi\' c \\ \small{UC Davis} \\ \small{smitrovic@ucdavis.edu}
}

\date{}
\begin{document}
\maketitle

\begin{abstract}
    This paper presents an $O(\log\log \bar{d})$ round massively parallel algorithm for $1+\epsilon$ approximation of maximum weighted $b$-matchings, using near-linear memory per machine. Here $\bar{d}$ denotes the average degree in the graph and $\epsilon$ is an arbitrarily small positive constant. Recall that $b$-matching is the natural and well-studied generalization of the matching problem where different vertices are allowed to have multiple (and differing number of) incident edges in the matching. Concretely, each vertex $v$ is given a positive integer budget $b_v$ and it can have up to $b_v$ incident edges in the matching.
    Previously, there were known algorithms with round complexity $O(\log\log n)$, or $O(\log\log \Delta)$ where $\Delta$ denotes maximum degree, for $1+\epsilon$ approximation of weighted matching and for maximal matching [Czumaj et al., STOC'18, Ghaffari et al. PODC'18; Assadi et al. SODA'19; Behnezhad et al. FOCS'19; Gamlath et al. PODC'19], but these algorithms do not extend to the more general $b$-matching problem. 
\end{abstract}

\section{Introduction}
Over the past few years, there has been significant progress in developing theoretically rigorous algorithms for large-scale graph problems, and particularly massively parallel algorithms for graph problems. One of the highlights has been the development of the \emph{round compression} technique that enables us to obtain massively parallel algorithms that are exponentially faster than their standard parallel/distributed counterparts, e.g. achieving an $O(\log\log n)$ round complexity, rather than the more standard $\log n$ or $\poly(\log n)$ round complexity in distributed or PRAM parallel settings. However, the applicability range of this technique has remained somewhat limited. In this paper, we contribute to broadening this range beyond its current boundary, to include the $b$-matching problem. We also develop novel augmentation techniques for the $b$-matching problem, which allow us to sharpen the approximation factor to $1+
\epsilon$, for both unweighted and weighted $b$-matching. To present these contributions in the proper context, we first review the model and the state of the art for matching approximation. We then state our results and provide an overview of the technical novelties involved.

\subsection{Massively Parallel Computation Model.} Sparked by the success of large-scale distributed and parallel computing platforms such as MapReduce ~\cite{Dean2004MapReduceSD}, Hadoop~\cite{White2009HadoopTD}, Dryad~\cite{Isard2007DryadDD}, and Spark~\cite{Zaharia2010SparkCC}, there has been increasing interest in developing theoretically sound algorithms for these settings. The Massively Parallel Computation (MPC) model has emerged as the de facto standard theoretical abstractions for parallel computation in such settings. This model was first introduced by Karloff et al.\ \cite{Karloff2010AMO} and Feldman et al.\ \cite{Feldman2008OnDS}, and was refined in several follow up work \cite{Goodrich2011SortingSA, beame2017communication, andoni2014parallel}. 

In the MPC model, when discussing a problem on an input graph $G=(V, E)$, we assume that the graph $G$, which has $n=|V|$ vertices, is partitioned among $M$ machines and each machine knows only some of the edges (and vertices). A key parameter of the model is the memory $S$ per machine. Since the machines should be able to hold the graph together, we have that $M\cdot S\geq (|V|+|E|)$, and it is common to assume that this is tight up to logarithmic factors, i.e., $M\cdot S = \tilde{O}(|V|+|E|)$. We usually refer to $M\cdot S$  as \emph{global memory} or \emph{total memory}, while the memory per machine $S$ is often called \emph{local memory}. 

Initially, the input consisting of the edges and vertices of the graph is divided arbitrarily among all machines, subject to the constraint that each holds $S$ words. Computation proceeds in synchronous rounds, where per round each machine can execute some (usually polynomial-time) computation on the data it holds. Afterward, there is a round of communication where each machine can send some data to each other machine -- thus the communication network among the machines graph is the complete graph. The only restriction on the communication is that the total amount of data that one machine sends and receives cannot exceed its local memory $S$. The main measure of interest is the number of rounds to solve the graph problem.

For many problems, the local memory parameter $S$ impacts the difficulty of the problem significantly --- problems get harder as we reduce $S$. Considering this, throughout the literature, the focus has been primarily on three regimes of this MPC: (A) \emph{strongly super-linear} memory regime, when $S=n^{1+c}$ for some positive constant $c>0$, (B) \emph{near-linear} memory regime when $S=\tilde{O}(n)$, and (C) \emph{strongly  sub-linear} memory regime, when $S=n^{1-c}$ for some positive constant $c>0$. Often the algorithms in the strongly super-linear regime or strongly sub-linear regime do not depend on the exact value of the constant $c$, and their complexity degrades by only a constant factor if we change $c$.

\subsection{State of the art for the matching problem} There are well-known distributed algorithms that compute a maximal matching and also $1+\epsilon$ approximation of maximum matching, for any positive constant $\epsilon$, in $O(\log n)$ rounds~\cite{israeli1986fast, lotker2015improved} of the $\mathsf{LOCAL}$ model of distributed computing.
These algorithms can be easily adapted to the MPC setting with the same round complexity (in either of the memory regimes). The focus in MPC has been on obtaining much faster algorithms. 

Lattanzi et al.~\cite{lattanzi2011filtering} presented a constant-round algorithm for maximal matching in the regime of strongly super-linear memory, and a constant-round algorithm for 8-approximate weighted maximum matching. Progress on lower memory regimes was much slower and no sub-logarithmic time algorithm was known,  until a breakthrough of Czumaj et al. \cite{czumaj2019round} that presented the first round compression. Their algorithm computes a $(1+\epsilon)$ approximation of a maximum matching in unweighted graphs in $O((\log \log n)^2)$ MPC rounds. Later this complexity was improved by Assadi et al.\ \cite{assadi2019coresets} and Ghaffari et al.\ \cite{ghaffari2018improved} to  $O(\log \log n)$ rounds. Gamlath et al.\ \cite{gamlath2019weighted} showed an extension to the weighted case of matching, resulting in a $(1+\epsilon)$-approximation algorithm for maximum weighted matching in $O(\log \log n)$ rounds for constant $\eps$. 
Round compression is a technique for ``approximately'' simulating certain algorithms with locality radius $R$ in much fewer than $R$ MPC rounds. At a high-level, an approximate simulation of algorithm $\cA$ is carried as follows. First, in $O(1)$ MPC rounds this technique distributes the input graph across machines. Second, each machine executes many steps of $\cA$, or of a slight modification of $\cA$, on its \emph{local} subgraph. In the known results and when the memory per machine is $\tO(n)$, on average, it is possible to approximately simulate $O(R / \log R)$ steps of $\cA$ per single MPC round. The simulation is performed in such a way that execution of $\cA$ on the subgraphs on different machines is almost the same as executing $\cA$ on the input graph.

Let us mention that there are a number of other works in this area that are less relevant for the focus of this paper. As sample examples, we mention the results on the dual problem of vertex cover approximation~\cite{assadi2019coresets, ghaffari2018improved, ghaffari2020massively}, those on maximal independent set~\cite{ghaffari2018improved},  maximal matching~\cite{behnezhad2019exponentially}, and coloring~\cite{chang2019complexity, czumaj2021improved, czumaj2021improved} problems. Moreover, there have been some improvements for the matching problem in the regime of strongly sublinear memory, e.g., \cite{ghaffari2019sparsifying, behnezhad2019massively}. However, these algorithms do not achieve a round complexity of $\poly(\log\log n)$ in general graphs, and thus they are not directly comparable to or relevant for the current paper.

\subsection{Our Contribution} In this paper, our focus is on the $b$-matching problem, which is a well-motivated and well-studied generalization of the matching problem: here each vertex $v$ is prescribed an integer budget $b_v$ and we call a subset of edges a $b$-matching if each vertex $v$ has at most $b_v$ edges in this subset. Notice that standard matching is the special case where $b_v=1$ for all vertices $v$.   Notice that while some of the typical examples for the maximum matching problem involve simple matchings between items (e.g., boys to girls, in the stable marriage problem), a wider range of the allocation problems are better captured by the $b$-matching problem where the entities have heterogeneous capacities. For instance, in client-to-server matching, often servers can serve a larger number of requests (and often a varying number, perhaps depending on the time) and even the clients might have different numbers of requests. 

To the best of our knowledge, none of the known MPC algorithms for the matching problem extend to the $b$-matching problem. We comment that an $O(\log n)$-round algorithm for $2$-approximation of the fractional relaxation of weighted $b$-matching problem is provided by~\cite{koufogiannakis2009distributed}, and we believe this can be easily turned into an $O(1)$ approximation for the integral case of $b$-matching\footnote{Technically, this algorithm applies to a different variant of $b$-matching where we are allowed to take each edge multiple times in the $b$-matching. We believe their algorithm can be turned into a $3$-approximation for the variant we discuss here where we are allowed to use each edge only once.}.
This algorithm can be run in the MPC model with the same $O(\log n)$ round complexity. However, we are not aware of any sublogarithmic-time MPC algorithm for $b$-matching, even when relaxing to $O(1)$ approximation of unweighted $b$-matching, and even when relaxing to the fractional variant.

We present an $O(\log\log n)$ round MPC algorithm for $(1+\epsilon)$ approximation of $b$-matching in weighted graphs in the near-linear memory regime.  
\begin{restatable}{theorem}{maintheoremweighted}\label{them:main-weighted-eps-constant}
There is a randomized MPC algorithm that computes a $(1+\epsilon)$ approximation of $b$-matching in weighted graphs in $O(\log\log \bar{d})$ rounds, using $\tilde{O}(n)$ local memory and $\tilde{O}(m+n)$ global memory. Here, $\bar{d}$ denotes the average degree of the graph and $\eps$ is assumed to be a small positive constant.
\end{restatable}
We note that the algorithm matches the performance of the best-known results for the simpler problem of matching (i.e., when $b_v=1$ for all vertices $v$) in unweighted graphs, in round complexity and the bounds on local and global memory for $1+\eps$ approximation. Even for an arbitrary constant approximation of matching in unweighted graphs, no faster algorithm is known for the near-linear regime of local memory.

    Our techniques for computing unweighted and weighted $(1+\eps)$-approximate $b$-matchings are relatively general, and they also yield semi-streaming constant-pass algorithms with $\tO(\sum_v b_v)$ memory.

\subsection{Method Overview}
\label{sec:overview-of-techniques}


\paragraph{$\Theta(1)$ approximation of unweighted $b$-matchings.} 
This result is presented in \cref{sec:Theta(1)-approximate} and
\cref{sec:tight_lp}.
Our algorithm for computing a $\Theta(1)$-approximate unweighted $b$-matching is a slight variation of the approach of \cite{ghaffari2020massively} for the approximate minimum weighted vertex cover problem, which itself can be seen as a generalization of the algorithm of \cite{ghaffari2018improved}. 

To explain our algorithm, we start with a review of the algorithm of \cite{ghaffari2018improved}, which computes a $\Theta(1)$-approximate matching. This algorithm starts by computing a $\Theta(1)$-approximate \emph{fractional} matching. By a simple sampling approach, this fractional matching can then be turned into an integral one, at the expense of losing an additional constant factor in the approximation guarantee.

The starting point for computing the fractional matching is the following $O(\log n)$-round procedure, which computes a sequence of fractional matchings with the final one being a constant approximation.
The first fractional matching assigns each edge a value of $\frac{1}{n}$. Now, consider a round of the process. We refer to a vertex as being \emph{active} during the round if the fractional values of the incident edges sum to at most $0.5$. We refer to an edge as being active if both of its endpoints are active. One now obtains the next fractional matching in the sequence by increasing the weight of each active edge by a factor of $2$.
A slightly improved version of this process assigns each edge at the beginning a value of $\frac{1}{\Delta}$. With this initialization, one obtains a constant approximation after $O(\log \Delta)$ rounds.

The main idea of the round compression technique is to simulate multiple rounds of this process in $O(1)$ MPC rounds. By simulation, we do not mean an exact simulation, but instead, the algorithm considers an approximate version of the aforementioned process. At the beginning, the algorithm starts by randomly partitioning the vertices across $\sqrt{\Delta}$ machines, and it stores in each machine the graph induced by its assigned vertices. A simple calculation shows that the expected number of edges assigned to a given machine is $O(n)$, and it is also straightforward to show concentration.
Now, during a given round, a vertex bases its decision on whether to stay active not on the values of \emph{all} of its incident edges; instead,  it computes an estimation of that sum based on the incident edges that were sent to the same machine.
We note that the idea of simulating multiple rounds of a distributed algorithm in $O(1)$ MPC rounds via vertex partitioning was first introduced in \cite{czumaj2019round}.

As vertices base their decision on whether to stay active or not only on local information, multiple rounds of this process can be simulated without communication between the machines. However, in each additional simulated round, the estimates become less and less precise, i.e., the difference between the exact and the approximate simulation becomes more and more apparent. One major reason the estimates become less precise is that the value assigned to a given edge might increase in each round, and therefore the impact of that edge on the vertex estimate. For this reason, it is only possible to simulate $c \cdot \log(\Delta)$ iterations, for some fixed constant $c < 1$, before certain concentration arguments start to break down.
In the idealized process, simulating that many rounds is sufficient to reduce the maximum degree of the graph induced by active vertices by a polynomial factor, i.e., to $\Delta^{0.999}$. 
In the approximate version, this also holds true, at least if one ignores the small number of vertices whose estimates were imprecise. 
This maximum degree reduction then allows for a more aggressive sampling rate in the following round compression step.
More generally, in each round compression step the maximum degree of the ``active'' graph drops by a polynomial factor. Hence, after $O(\log\log \Delta)$ round compression steps the maximum degree is $\poly(\log n)$. The remaining instance can then either be solved on a single machine, resulting in an overall round complexity of $O(\log\log \Delta)$, or by direct simulation  of an $O(\log \log n)$-round distributed algorithm using \emph{sublinear} local memory, resulting in an overall round complexity of $O(\log \log n)$.

We next discuss our algorithm. The main ideas discussed below were already introduced in the approximate weighted minimum vertex cover algorithm of \cite{ghaffari2020massively}, running in $O(\log \log \bar{d})$ MPC rounds.
The algorithm of \cite{ghaffari2020massively} computes a dual solution of the fractional relaxation of weighted minimum vertex cover. The dual problem is a relaxation of the $b$-matching variant where a given edge can be chosen multiple times, instead of just once.
While it is not apparent how to use their algorithm in a black-box manner, we show that it is possible to adapt their algorithm to compute a fractional $b$-matching, and we provide a self-contained proof in  \cref{sec:Theta(1)-approximate} and 
\cref{sec:tight_lp}.

Our algorithm also starts by computing a fractional solution. The baseline is essentially the same procedure as described above, with three differences. First, the initial fractional $b$-matching is defined differently. Second, a vertex considers itself active only if the fractional values of the incident edges sum to at most $0.5b_v$ instead of to at most $0.5$. Third, an edge also stops being active if its fractional value exceeds $0.5$. 

Note that the initial $b$-matching needs to balance two things. If the assigned edge values are too low, then simulating multiple rounds of the sequential process does not lead to any intermediate progress, such as a substantial maximum-degree reduction.
On the other hand, if the assigned values are too large, then in the extreme case they do not even constitute a valid fractional $b$-matching and, even if they do, the computed estimates during the round compression might be too imprecise due to a large influence of individual edges on vertex estimates.
One initialization that works is $\min(\frac{b_v}{\max(\bar{d},d_v)},\frac{b_u}{\max(\bar{d},d_u)},1)$.
This initialization results in a valid fractional $b$-matching. We note that replacing $\max(\bar{d},d_v)$ with $d_v$ (and the same for $u$), would still result in a valid fractional $b$-matching, however, the values assigned to edges incident to low-degree vertices would be too large to obtain accurate enough estimates.
This initialization does not lead to intermediate progress in the form of a maximum degree reduction. It is however possible to show that the average degree drops by analyzing the out-degree of each node in the directed graph that one obtains by directing each edge $\{u,v\}$ from $u$ to $v$ if $b_u < b_v$ (and in an arbitrary direction if $b_u = b_v$).

\paragraph{$(1+\eps)$-approximate unweighted $b$-matching.}

This result is presented in \cref{sec:1+eps-unweighted}.
The first step of our approach for obtaining better than $O(1)$-approximate $b$-matching is to characterize the structure of augmenting walks. One can show that considering only augmenting paths, as it is possible for $1$-matchings, is not always sufficient to improve the current $b$-matching. To obtain a characterization, we propose a view of $b$-matchings that essentially enables us to carry over existing properties for $1$-matching. Namely, we consider a graph $\tG$ obtained by copying each vertex $v$ $b_v$ many times. Then, we show in \cref{sec:existence-of-short-augmentations} that, if the current $b$-matching is not maximum, there exists a collection of augmenting \emph{paths} in $\tG$ that improve the current $b$-matching. To our knowledge, this connection has not been used in the literature before.

Second, we would like to build on our connection between $1$- and $b$-matchings and reuse existing results for obtaining $(1+\eps)$-approximate $1$-matchings. To that end, we consider the algorithm of \cite{mcgregor2005finding}, that can also be executed efficiently in MPC. However, reusing that result in the context of $b$-matchings and $\tG$ has a major challenge: when we create $b_v$ copies of each vertex $v$, it is not clear between which copy of $u$ and which copy of $v$ the edge $e = \{u, v\}$ should be placed. It is not hard to show that if $e$ is matched, then it can be placed between arbitrary copies as long as a copy is incident to at most one matched edge. However, this situation is significantly more complicated if $e$ is unmatched. Nevertheless, we show that there is a way to go back to $b'$-matchings from a \emph{subset} of vertices of $\tG$ -- this subset is algorithmically defined via the approach of \cite{mcgregor2005finding} -- that enables us to avoid assigning $e$ between specific copies a-priori.

\paragraph{$(1+\eps)$-approximate weighted $b$-matching.}
This result is presented in \cref{sec:weighted-b-matching}.  Our starting point is a result of \cite{gamlath2019weighted} that proves \cref{them:main-weighted-eps-constant} in the special case when $b = 1$. There are two major challenges in using that framework for $b$-matchings. First, it is not clear how to use that result when a vertex can be incident to multiple matching edges. Second, even if we could apply it to $b$-matchings, there is a step that requires resolving conflicts between a collection of alternating walks. That step in the prior work requires $O(|M|)$ space per machine, which in the case of $b$-matchings can be much larger than $O(n)$.

The main challenge in applying the framework to $b$-matchings is that using it directly might result in an alternating walk that crosses the same edge multiple times -- such alternating walks are not proper augmentations. We resolve that by adding an additional rule to the framework. Essentially, we assign random orientations to unmatched edges enabling us to avoid described behavior. This process is explained as Step~(III) in \cref{section:layering-for-b-matchings}.

At three steps, the algorithms in \cite{gamlath2019weighted} in a crucial way require the memory per machine to be $O(|\Mopt| \cdot \exp(1/\eps)) \gg O(n \cdot \exp(1/\eps))$, where $\Mopt$ is a maximum matching and in case of $b$-matchings can have size $\sum_v b_v$. The main reason for that is that the framework finds a collection of augmentation which might intersect. Then, the framework chooses an independent set of them, the process which we call \emph{conflict resolution}. To perform conflict resolution in \cite{gamlath2019weighted} all augmentations are collected to one machine and a maximal independent set of them is chosen greedily. We design a different conflict resolution scheme that enables us to select an independent set (not necessarily maximal) of augmentations in a fully scalable manner, requiring only $O(n^\delta + \poly(1/\eps))$ memory per machine for any arbitrary constant $\delta > 0$.

\paragraph{Streaming}
    A key challenge in obtaining a semi-streaming implementation of our approach is choosing the aforementioned orientation of unmatched edges. Namely, the chosen orientation has to be ``remembered'' for multiple passes, and doing that directly for all unmatched edges requires $O(m) \gg O(\sum_v b_v)$ memory. Nevertheless, we show that this choice of orientations requires only $O(\poly(1/\eps))$ independence, which enables us to use $k$-wise independent hash functions (for small $k$), resulting in a memory-efficient semi-streaming implementation.
\section{Preliminaries}
\begin{definition}[$b$-matching]
    Let $G = (V, E)$ be a graph, $b \in \bbN_{\ge 1}^V$ a vector and $M$ a set of edges. We say that $M$ is a $b$-matching if for each vertex $v$ there are at most $b_v$ edges in $M$ incident to $v$. To refer to a $1$-matching, we also say \emph{matching}.
\end{definition}

\begin{definition}[Walk]
    Let $G = (V, E)$ be a graph. We use \emph{walk} to refer to a sequence $P$ of vertices (and the corresponding edges) where every two neighboring vertices in $P$ are adjacent in $G$. It is allowed that vertices in $P$ repeat.
\end{definition}

\begin{definition}[Alternating walks]
    Let $G = (V, E)$ be a graph, $M$ a set of edges and $P$ a walk in $G$. We say that $P$ is an alternating walk if its edges alternate between those in $E \setminus M$ and $M$. The first edge of $P$ does not have to be in $E \setminus M$.
\end{definition}

\begin{definition}[Free vertex]
    Given a $b$-matching $M$ and a graph $G$, a vertex $v \in V(G)$ is called \emph{free} with respect to $M$ if the number of edges in $M$ incident to $v$ is less than $b_v$. When $M$ is a $1$-matching, this definition implies that $v$ is free if it is unmatched in $M$.
\end{definition}

\section{$\Theta(1)$ Approximation of Unweighted $b$-Matchings}
\label[section]{sec:Theta(1)-approximate}

We first provide an outline of how we prove the following theorem:

\begin{theorem}
\label{thm:main-O(1)-approximate}
There exists an $O(\log \log \bar{d})$-round MPC algorithm using $\tilde{O}(n)$ local memory and $O(m + n)$ global memory which computes a $\Theta(1)$ approximation of $b$-matching with positive constant probability.
\end{theorem}

\subsection{Outline}
Note that one can obtain a ``with high probability'' guarantee in \cref{thm:main-O(1)-approximate} at the expense of increasing the global memory by an $O(\log n)$-factor by running the algorithm $O(\log n)$ times in parallel and outputting the largest computed $b$-matching.

As written in the introduction, the first part of our algorithm consists of finding a fractional $b$-matching solution.
In fact, our algorithm can find a constant factor approximation to the following LP:

	\[
	\begin{array}{lcll}
   &\text{maximize} &\sum_{e \in E} x_e  \\
   &\text{subject to} &  \sum_{e \in E(v)} x_e \leq b_v &\text{ for every $v \in V $} \\ 
   &                  & x_e \leq r_e &\text{ for every $e \in E$} \\
   & &x \ge 0,
\end{array}
\]
where we denote by $E(v)$ the set of edges incident to $v$.
For solving $b$-matching, one can assume that all the $b_v$'s are non-negative integers and $r_e = 1$ for every edge $e$. However, our algorithm works in the more general setting where the $b_v$'s and $r_e$'s can be arbitrary non-negative reals.

We note that the algorithm of \cite{ghaffari2020massively} finds a constant approximation to the LP one obtains by dropping the edge constraints $x_e \leq r_e$.

Our algorithm computes a feasible $0.05$-tight solution $x$, a notion defined below.
\begin{definition}[$\alpha$-tight, $\alpha$-loose]
Consider an arbitrary $x \in \mathbb{R}^E_{\geq 0}$ and $\alpha \in [0,1]$. 
We define 

\[V^{loose}(x,\alpha) = \{v \in V \colon \sum_{e \in E(v)} x_e < \alpha b_v\}\]

and 

\[E^{loose}(x,\alpha) = \{e \in E \colon x_e < \alpha r_e\} \cap \binom{V^{loose}(x,\alpha)}{2}.\]

If $E^{loose}(x,\alpha) = \emptyset$, then we refer to $x$ as being $\alpha$-tight.
\end{definition}

Let $OPT$ denote the optimal value of the LP.
By using duality, one can show that an $\alpha$-tight solution $x$ has an objective value of at least $\frac{\alpha}{3}OPT$. By setting $r_e = 1$ for every $e \in E$, the LP is a relaxation of $b$-matching and therefore an $\alpha$-tight solution $x$ satisfies that $\sum_{e \in E} x_e$ is at most a $\frac{3}{\alpha}$-factor away from the maximum $b$-matching size. Now, consider sampling each edge with probability $\frac{x_e}{4}$ and removing all edges incident to vertices with too many sampled edges. It is straightforward to show that this sampling process results in a $b$-matching whose expected size is within a constant factor of $\sum_{e \in E} x_e$. Combining the above observations, one obtains the following lemma, whose proof is deferred to \cref{sec:matching_from_tight}.

\begin{lemma}
\label{lem:matching_from_tight_solution}
Consider an arbitrary $\alpha \in (0,1]$ and let $x \in \mathbb{R}^E_{\geq 0}$ be a feasible $\alpha$-tight solution of the LP with $r_e$ being set to $1$ for every $e \in E$. Then, we can compute with positive constant probability in $O(1)$ MPC rounds with $O(n)$ local and $O(n + m)$ global an $O(1/\alpha)$-approximate $b$-matching $M$.
\end{lemma}

Thus, it remains to provide an efficient algorithm to compute a $0.05$-tight feasible solution.
 As mentioned before, our algorithm for computing such a tight solution is very similar to the algorithm of \cite{ghaffari2020massively}.
 We nevertheless provide a complete description and analysis of the algorithm in \cref{sec:tight_lp}, as our algorithm does not follow from a black-box reduction of \cite{ghaffari2020massively}.

\subsection{Detailed analysis}
\label{sec:tight_lp}

This section is dedicated to giving an efficient MPC algorithm for computing a $0.05$-tight feasible solution of the LP introduced in \cref{sec:Theta(1)-approximate}.

As mentioned before, the algorithm and analysis is very similar to the one given in \cite{ghaffari2018improved, ghaffari2020massively}.

The algorithm description and analysis consists of three parts.
In this brief informal overview, we consider the special case of the LP with $r_e = 1$ for every edge $e$. As mentioned earlier, the algorithms work for the more general case.
In \cref{sec:sequential}, we start by giving a formal description of the idealized process that computes a sequence of fractional $b$-matchings.
For a given $T$, we show that the fractional $b$-matching $x$ that one obtains after $T$ rounds satisfies $|E^{loose}(x,0.2)| \leq \frac{5|E|}{2^T}$, i.e., the number of edges one still needs to consider to obtain a $0.2$-tight solution starting from $x$ drops exponentially in $T$.

We note that the idealized process we describe slightly differs from the one in the introduction. First, the value assigned to each edge at the beginning is a constant factor lower. Second, a vertex does not decide whether to stay active based on whether the sum of incident edge values exceeds $0.5b_v$. Instead, it chooses in each round an independent random threshold $\mathcal{T}$ chosen uniformly at random from the interval $[0.2b_v,0.4b_v]$ and checks whether the sum of edge values exceeds that random threshold.
While these two minor modifications are not relevant for showing correctness or the edge reduction property, they become important for the analysis in \cref{sec:MPC_simulation}. 
In \cref{sec:MPC_simulation}, we describe an algorithm that approximates the idealized version, along the lines of what has been discussed in the introduction, and outputs the fractional $b$-matching obtained after $T = c \log(\bar{d})$ rounds, where $c$ is a sufficiently small constant. The main technical part of this section is dedicated to show that the computed fractional $b$-matching $\tilde{x}$ satisfies  $\E{\frac{|E^{loose}(\tilde{x},0.05)|}{n}} \leq \left(\frac{m}{n}\right)^{0.9999}$, i.e., the average degree drops by a polynomial factor, in expectation.
Note that the fractional $b$-matching $x$ one obtains after $c \log(\bar{d})$ rounds of the idealized process satisfies $\frac{|E^{loose}(x,0.2)|}{n} \leq \left(\frac{m}{n}\right)^{0.99}$. The analysis therefore proceeds by comparing the idealized process with the approximate version. In particular, we consider the coupling obtained by choosing the same random thresholds in both executions and show that  the probability of a given edge being contained in $E^{loose}(\tilde{x},0.05) \setminus E^{loose}(x,0.2)$ is sufficiently small in order to obtain, together with $\frac{|E^{loose}(x,0.2)|}{n} \leq \left(\frac{m}{n}\right)^{0.99}$, that $\E{\frac{|E^{loose}(\tilde{x},0.05)|}{n}} \leq \left(\frac{m}{n}\right)^{0.9999}$.

Finally, we show in \cref{sec:putting_evertyhing_together} how to compute a $0.05$-tight feasible solution by performing $O(\log \log \bar{d})$ invocations of the average degree reduction algorithm of \cref{sec:MPC_simulation}, which runs in $O(1)$ MPC rounds.

\subsection{Idealized Process}
\label{sec:sequential}

Starting from now, we consider an arbitrary $b \in \mathbb{R}^V_{\geq  0}$ and $r \in \mathbb{R}^E_{\geq 0}$. \cref{alg:idealized} gives a formal description of the idealized process.

\begin{algorithm}[h]
\TitleOfAlgo{$Sequential(G,b,r,T)$}
    \SetAlgoLined
    \KwData{$G=(V,E)$ is an unweighted graph, $b \in \mathbb{R}^V_{\geq 0}$, $r \in \mathbb{R}^{E}_{\geq 0}$, $T \in \mathbb{N}$}
    \KwResult{$x \in \mathbb{R}^E_{\geq 0}$}
    $V_0^{active} = V$\;
    $\forall v \in V, t \in [T]$, $\mathcal{T}_{v,t} \leftarrow \mathcal{U}(0.2b_v,0.4b_v)$\;
    $\forall v \in V \colon q_v = 0.8\frac{b_v}{\max(|E(v)|,\bar{d})} $\;
    $\forall e = \{u,v\} \in E \colon x_{e,0} = \min \left(r_e, q_v, q_u\right)$ \;

    \For{$t = 1,2,\ldots,T$}{
         $ \forall v \in V \colon y_{v,t-1} = \sum_{e \in E(v)} x_{e,t-1}$\;
        $V_t^{active} = \{v \in V^{active}_{t-1} \colon y_{v,t-1} \leq \mathcal{T}_{v,t}\}$ \;
         $E_t^{active} = \{e \in E \cap \binom{V^{active}_t}{2} \colon x_{e,t-1} \leq r_e/2\}$ \;
        $ \forall e \in E \colon x_{e,t} = \left\{
\begin{array}{ll}
2x_{e,t-1}  &\text{, if $e \in E^{active}_t$} \\
x_{e,t-1} & \text{, otherwise} \\
\end{array}
\right.$ \;

    }
     $\forall e \in E \colon x_e = x_{e,T}$\;
    \Return{$x$}
\caption{\label{alg:idealized} Idealized process running for $T$ rounds}

\end{algorithm}

Note that the returned vector $x$ is a feasible LP solution. In particular, the following lemma follows from a simple induction.

\begin{lemma}
\label{lem:primal_feasible}
For each $t \in \{0,\ldots,T\}$, we have $\sum_{e \in E(v)} x_{e,t} \leq 0.8b_v$ for every $v \in V$ and $0 \leq x_{e,t} \leq r_e$ for every $e \in E$.
\end{lemma}
 Note that for showing feasibility, it would be sufficient to have $\sum_{e \in E(v)} x_{e,t} \leq b_v$ instead of $\sum_{e \in E(v)} x_{e,t} \leq 0.8b_v$. This stronger upper bound is used during the analysis of the coupling of the idealized and approximate process.

We next verify that the number of edges in $E^{loose}(x,0.2)$ indeed decreases exponentially with $T$. 

\begin{lemma}
\label{lem:loose_edges}
Let $x = Sequential(G,b,r,T)$. Then,
 $|E^{loose}(x,0.2)| \leq \frac{5|E|}{2^T}$.
\end{lemma}
\begin{proof}
For every $v \in V$, we define $\overrightarrow{E}(v) = \{e \in E(v) \colon x_{e,0} = q_v\}$.
Note that for every $e \in \overrightarrow{E}(v) \cap E^{loose}(x,0.2)$, it holds that $x_e = q_v \cdot 2^T$. Hence, we obtain

\[|\overrightarrow{E}(v) \cap E^{loose}(x,0.2)| \leq \frac{b_v}{q_v \cdot 2^T} \leq \frac{\max(|E(v)|,\bar{d})}{0.8 \cdot 2^T}.\]

Moreover, for every $e = \{u,v\} \in E^{loose}(x,0.2)$, $x_e = q_u$ or $x_e = q_v$ and therefore $e \in \overrightarrow{E}(u) \cup \overrightarrow{E}(v)$. Therefore,

\begin{align*}
    |E^{loose}(x,0.2)| \leq \sum_{v \in V} |\overrightarrow{E}(v) \cap E^{loose}(x,0.2)| 
    \leq \frac{1}{0.8 \cdot 2^T}\left(\sum_{v \in V}|E(v)| + \sum_{v \in V} \bar{d}\right) = \frac{2|E| + 2|E|}{0.8 \cdot 2^T} = \frac{5|E|}{2^T}.
\end{align*}
\end{proof}

\begin{theorem}
Let $x = Sequential(G,b,r, \lceil \log(5|E| + 1) \rceil)$. Then, $x$ is a $0.2$-tight primal feasible solution.
\end{theorem}
\begin{proof}
Feasibility directly follows from \cref{lem:primal_feasible}. Moreover, according to \cref{lem:loose_edges}, we have

\[|E^{loose}(x,0.2)| \leq \frac{5|E|}{2^{\lceil \log(5|E| + 1) \rceil}} < 1\]
and therefore $|E^{loose}(x,0.2)| = 0$. Hence, $x$ is indeed a $0.2$-tight primal feasible solution.
\end{proof}

\subsection{MPC Simulation}
\label{sec:MPC_simulation}
\cref{alg:one_round_mpc} gives a formal description of the approximate process we previously talked about. In this section and the next section, we implicitly assume that the $n$-vertex and $m$-edge input graph $G$ satisfies that $n$ is sufficiently large and $m \geq n\log^{10}(n)$.

\begin{algorithm}[h]
\TitleOfAlgo{$OneRoundMPC(G,b,r)$}

    \SetAlgoLined
    \KwData{$G=(V,E)$ is an unweighted graph, $b \in \mathbb{R}^V_{\geq 0}$, $r \in \mathbb{R}^{E}_{\geq 0}$}
    \KwResult{$x \in \mathbb{R}^E_{\geq 0}$}
    $N =\left\lceil \sqrt{\bar{d}} \right\rceil, T = \left \lfloor \log_2(N)/1000 \right \rfloor, \tilde{V}_0^{active} = V$\;
    Every vertex $v \in V$ independently and uniformly at random chooses an index $i_v \in [N]$ \;
    $ \forall i \in [N] \colon V_i = \{v \in V \colon i_v = i\},  E^{local}(v) = E(v)  \cap \binom{V_{i_v}}{2}$ \;
    $\forall e = \{u,v\} \in E \colon \tilde{x}_{e,0} = x_{e,0}$ \;

    \For{$t = 1,2,\ldots,T$}{
         $ \forall v \in V \colon \tilde{y}_{v,t-1} = N \cdot \sum_{e \in E^{local}(v)} \tilde{x}_{e,t-1}$\; \label{line:estimate}
        $\tilde{V}_t^{active} = \{v \in \tilde{V}^{active}_{t-1} \colon \tilde{y}_{v,t-1} \leq \mathcal{T}_{v,t}\}$ \;
         $\tilde{E}_t^{active} = \{e \in E \cap \binom{\tilde{V}^{active}_t}{2} \colon \tilde{x}_{e,t-1} \leq r_e/2\}$ \;
        $ \forall e \in E \colon \tilde{x}_{e,t} = \left\{
\begin{array}{ll}
2\tilde{x}_{e,t-1}  &\text{, if $e \in \tilde{E}^{active}_t$} \\
\tilde{x}_{e,t-1} & \text{, otherwise} \\
\end{array}
\right.$ \;

    }
    $ \forall  e = \{u,v\} \in E \colon \tilde{x}_e = \left\{
\begin{array}{ll}
\tilde{x}_{e,T}  &\text{, if $\sum_{e \in E(v)} \tilde{x}_{e,T} \leq b_v$ and $\sum_{e \in E(u)} \tilde{x}_{e,T} \leq b_u$ } \\
0 & \text{, otherwise} \\
\end{array}\right. $ \label{line:ensure_feasibility} \;
    \Return{$\tilde{x}$}
    \caption{  Approximate Process  \label{alg:one_round_mpc}} 

\end{algorithm}

$OneRoundMPC(G,b,r)$ tries to imitate the algorithm $Sequential(G,b,r,T)$ with $T$ being set to $\lfloor \log_2(N)/1000) \rfloor$. Here, $N = \lceil \sqrt{\bar{d}}\rceil$ refers to the number of machines involved in the random partitioning step during the MPC simulation.
The one major difference in comparison to \cref{alg:idealized} can be found on \cref{line:estimate}. Instead of computing $\tilde{y}_{v,t-1} =\sum_{e \in E(v)} \tilde{x}_{e,t-1}$, as is done in $Sequential(G,b,r,T)$, the algorithm instead only uses the estimate $N \cdot \sum_{e \in E^{local}(v)} \tilde{x}_{e,t-1}$, where $E^{local}(v)$ denotes the set consisting of those edges incident to $v$ that are sent to the same machine as $v$.
Right at the beginning, it holds that $\E{N \cdot \sum_{e \in E^{local}(v)} \tilde{x}_{e,0}} = \sum_{e \in E(v)} \tilde{x}_{e,0}$.
To show concentration, it suffices that $\tilde{x}_{e,0}$ is sufficiently small for each edge $e \in E(v)$. Similarly, to show that $N \cdot \sum_{e \in E^{local}(v)} \tilde{x}_{e,t}$ is close to $\sum_{e \in E(v)} \tilde{x}_{e,t}$ for some $t > 0$ (and most vertices $v$), one similarly has to ensure that $\tilde{x}_{e,t}$ is sufficiently small. Here the following bounds suffice.
\begin{lemma}
\label{lem:maximum_weight}
Consider an arbitrary $v \in V$ and edge $e$ incident to $v$.
Then, for every $t \in \{0,1,\ldots,T\}$, $\tilde{x}_{e,t} \leq \tilde{x}_{e,0} \cdot N^{0.001} \leq \frac{b_v}{N^{1.99}}$ and $x_{e,t} \leq x_{e,0} \cdot N^{0.001} \leq \frac{b_v}{N^{1.99}}$.
\end{lemma}
\begin{proof}
We have 
\[\tilde{x}_{e,t} \leq \tilde{x}_{e,0} \cdot 2^T \leq \tilde{x}_{e,0} \cdot N^{0.001} \leq \frac{0.8b_v N^{0.001}}{\Delta_{avg}} \leq \frac{b_v N^{0.001}}{N^2} \leq \frac{b_v}{N^{1.99}},\]

and the same holds for $x_{e,t}$.
\end{proof}
However, just having an upper bound on $\tilde{x}_{e,t}$ is not sufficient to argue that $N \cdot \sum_{e \in E^{local}(v)} \tilde{x}_{e,t}$ is close to $\sum_{e \in E(v)} \tilde{x}_{e,t}$ most of the times, but instead a more elaborate inductive argument is needed. 
The inductive analysis couples the execution of \cref{alg:idealized} (setting $T = \lfloor \log_2(N)/1000\rfloor$) and \cref{alg:one_round_mpc} by using the same randomness for the random thresholds (All future statements are with respect to that coupling). In particular, one obtains the following result, whose proof is given in \cref{sec:inductive_argument}.

\begin{lemma}
\label{lem:vertex_bad}
For every $v \in V$, $Pr[|\sum_{e \in E(v)} x_{e,T} - \sum_{e \in E(v)} \tilde{x}_{e,T}| > 0.1b_v] \leq \frac{1}{N^{0.1}}$.
\end{lemma}

Let $\tilde{x}^{(T)} \in \mathbb{R}^E$ be defined as the vector with $\tilde{x}^T_e = \tilde{x}_{e,T}$ for every $e \in E$ and recall that $x = Sequential(G,b,r,T)$.
\cref{lem:vertex_bad} directly implies the following lemma.
\begin{corollary}
\label{cor:bad_edge}
For every $e \in E$, $Pr[e \in E^{loose}(\tilde{x}^{(T)},0.1) \setminus E^{loose}(x,0.2)] \leq \frac{2}{N^{0.1}}$.
\end{corollary}
\begin{proof}
Consider an arbitrary edge $e = \{u,v\} \in E^{loose}(\tilde{x}^{(T)},0.1) \setminus E^{loose}(x,0.2)$. 
As $e \in E^{loose}(\tilde{x}^{(T)},0.1)$, we have

\[x_e \leq x_{e,0} \cdot 2^T = \tilde{x}_{e,T} < 0.1r_e \leq 0.2r_e.\]
The fact that $x_e < 0.2r_e$ together with $e \notin E^{loose}(x,0.2)$ implies that $u \notin V^{loose}(x,0.2)$ or $v \notin V^{loose}(x,0.2)$. 
As $u,v \in V^{loose}(\tilde{x}^{(T)},0.1)$, we therefore also get that $u \in V^{loose}(\tilde{x}^{(T)},0.1) \setminus V^{loose}(x,0.2)$ or $v \in V^{loose}(\tilde{x}^{(T)},0.1) \setminus V^{loose}(x,0.2)$.
If $u \in V^{loose}(\tilde{x}^{(T)},0.1) \setminus V^{loose}(x,0.2)$, then $|\sum_{e \in E(u)} x_{e,T} - \sum_{e \in E(u)} \tilde{x}_{e,T}| > 0.1b_u$, which happens with probability at most $\frac{1}{N^{0.1}}$ according to \cref{lem:vertex_bad}. Similarly, if $v \in V^{loose}(\tilde{x}^{(T)},0.1) \setminus V^{loose}(x,0.2)$, then $|\sum_{e \in E(v)} x_{e,T} - \sum_{e \in E(v)} \tilde{x}_{e,T}| > 0.1b_v$, which also happens with probability at most $\frac{1}{N^{0.1}}$. The statement therefore follows by a union bound.
\end{proof}

By using \cref{lem:loose_edges} together with \cref{cor:bad_edge}, it is already straightforward to show that the expected size of $E^{loose}(\tilde{x}^{(T)},0.1)$ is sufficiently small, i.e., that the average degree drops by a polynomial factor.
However, note that because vertices decide to be active based on estimates, $\tilde{x}^{(T)}$ might not even be a feasible solution, i.e., it is possible that there exists a vertex $v$ with $\sum_{e \in E(v)} \tilde{x}_{e,T} > b_v$. We refer to such a vertex as a bad vertex.
To ensure feasibility of the returned vector $\tilde{x}$, one obtains $\tilde{x}$ from $\tilde{x}^{(T)}$ by setting the value of all edges incident to bad vertices to $0$.
Luckily, it follows as a simple corollary of \cref{lem:vertex_bad} that a given vertex is only bad with a sufficiently small probability.
\begin{corollary}
\label{cor:bad_vertex}
For every $v \in V, Pr[\sum_{e \in E(v)}\tilde{x}_{e,T} > b_v] \leq \frac{1}{N^{0.1}}$.
\end{corollary}
\begin{proof}
\cref{lem:primal_feasible} states that $\sum_{e \in E(v)}x_{e,T} \leq 0.8b_v$. Hence, $\sum_{e \in E(v)}\tilde{x}_{e,T} > b_v$ implies $|\sum_{e \in E(v)} x_{e,T} - \sum_{e \in E(v)} \tilde{x}_{e,T}| > 0.1b_v$, which happens with probability at most $\frac{1}{N^{0.1}}$ according to \cref{lem:vertex_bad}. 
\end{proof}
On first sight, it might look that one could directly use \cref{cor:bad_vertex} to argue that $E^{loose}(\tilde{x},0.1) \setminus E^{loose}(\tilde{x}^{(T)},0.1)$ is sufficiently small. However, this is not the case as $E^{loose}(\tilde{x},0.1) \setminus E^{loose}(\tilde{x}^{(T)},0.1)$ can even contain edges that are not incident to a bad vertex, i.e., it can happen that a good vertex $v$ is contained in $V^{loose}(\tilde{x},0.1) \setminus V^{loose}(\tilde{x}^{(T)},0.1)$ if $v$ is neighboring with one or more bad vertices. We next show that the probability of $v$ being contained in $V^{loose}(\tilde{x},0.05) \setminus V^{loose}(\tilde{x}^{(T)},0.1)$ is sufficiently small. The reason for this is that $v$ can only be in $V^{loose}(\tilde{x},0.05) \setminus V^{loose}(\tilde{x}^{(T)},0.1)$ if a lot of incident edges are bad, where we denote an edge as bad if one of its endpoints is a bad vertex.

\begin{lemma}
\label{lem:vbad}
For every $v \in V$, $Pr[v \in V^{loose}(\tilde{x},0.05) \setminus V^{loose}(\tilde{x}^{(T)},0.1)] \leq \frac{1}{N^{0.08}}$
\end{lemma}
\begin{proof}
Consider some arbitrary vertex $v \in V$.
For each edge $e \in E(v)$, we set $X_e = \tilde{x}_{e,T} - \tilde{x}_e$.
Note that $X_e = \tilde{x}_{e,T}$ if $e$ is a bad edge and $X_e = 0$ otherwise.
According to \cref{lem:maximum_weight}, we have $\tilde{x}_{e,T} \leq \tilde{x}_{e,0} \cdot N^{0.001}$.
A given edge is bad if one of its endpoints is bad, which happens with probability at most $\frac{2}{N^{0.1}}$ according to \cref{cor:bad_vertex}.
Therefore, $\E{X_e} \leq \tilde{x}_{e,0} \cdot \frac{2N^{0.001}}{N^{0.1}} \leq \frac{1}{N^{0.09}}$.
Let $X = \sum_{e \in E(v)} X_e$. We have

\[\E{X}= \frac{1}{N^{0.09}} \cdot \sum_{e \in E(v)} \tilde{x}_{e,0} \leq \frac{b_v}{N^{0.09}}.\]

Hence, a Markov Bound implies that $Pr[X > 0.05b_v] \leq \frac{1}{N^{0.08}}$.
On the other hand, if $X \leq 0.05b_v$, then $\sum_{e \in E(v)} \tilde{x}_{e,T} - \sum_{e \in E(v)} \tilde{x}_e \leq 0.05b_v$, which directly implies that $v \notin V^{loose}(\tilde{x},0.05) \setminus V^{loose}(\tilde{x}^{(T)},0.1)$.
\end{proof}

\begin{lemma}
\label{lem:bad_edge}
For every $e \in E$, $Pr[e \in E^{loose}(\tilde{x},0.05) \setminus E^{loose}(x,0.2)] \leq \frac{4}{N^{0.08}}$.
\end{lemma}
\begin{proof}
Note that $e = \{u,v\} \in E^{loose}(\tilde{x},0.05) \setminus E^{loose}(x,0.2)$ implies $u \in V^{loose}(\tilde{x},0.05) \setminus V^{loose}(x,0.2)$ or $v \in V^{loose}(\tilde{x},0.05) \setminus V^{loose}(x,0.2)$.
Moreover, $u \in V^{loose}(\tilde{x},0.05) \setminus V^{loose}(x,0.2)$ implies $u \in V^{loose}(\tilde{x},0.05) \setminus V^{loose}(\tilde{x}^T,0.1)$ or $u \in V^{loose}(\tilde{x}^T,0.1) \setminus V^{loose}(x,0.2)$. The former happens with probability at most $\frac{1}{N^{0.08}}$ according to \cref{lem:vbad} and the latter happens with probability at most $\frac{1}{N^{0.1}}$ according to $\cref{lem:vertex_bad}$. By symmetry, the same holds for $v$ and therefore the lemma follows by a simple union bound.
\end{proof}

\begin{theorem}
\label{thm:expected_size}
We have $\E{\frac{|E^{loose}(\tilde{x},0.05)|}{n}} \leq \left(\frac{|E|}{n}\right)^{0.9999}$.
\end{theorem}
\begin{proof}

We have

\[\E{|E^{loose}(\tilde{x},0.05)|} \leq |E^{loose}(x,0.2)| + \E{|E^{loose}(\tilde{x},0.05) \setminus E^{loose}(x,0.2)| } .\]
According to \cref{lem:loose_edges}, we have

\[|E^{loose}(x,0.2)| \leq \frac{10|E|}{2^T} \leq \frac{20|E|}{N^{1/1000}}.\]

According to \cref{lem:bad_edge}, we have

\[\E{|E^{loose}(\tilde{x},0.05) \setminus E^{loose}(x,0.2)|} \leq \frac{4|E|}{N^{0.08}}.\]

Therefore,

\[\E{\frac{|E^{loose}(\tilde{x},0.05)|}{n}} \leq \frac{|E|}{n \cdot N^{0.0002}} \leq \frac{|E|}{n \cdot \bar{d}^{0.0001}} = \frac{|E|}{n \cdot (2|E|/n)^{0.0001}} \leq \left( \frac{|E|}{n}\right)^{0.9999}.\]

\end{proof}

\begin{theorem}
 Let $\tilde{x} = OneRoundMPC(G,b,r)$. Then, $\tilde{x}$ is a feasible solution of the LP and $\E{ \frac{|E^{loose}(\tilde{x},0.05)|}{n}} \leq  \left(\frac{m}{n}\right)^{0.9999}$. Moreover, $\tilde{x}$ can be computed in $O(1)$ rounds in MPC with local memory $\tilde{O}(n)$ and global memory $O(m)$.

\end{theorem}
\begin{proof}
The upper bound on the expected size of $E^{loose}(\tilde{x},0.05)$ directly follows from \cref{thm:expected_size} and it is also straightforward to see that $\tilde{x}$ is indeed a feasible solution of the LP.
Regarding the MPC implementation, 
\cref{lem:local_memory} shows by a simple concentration argument that $G[V_i]$ with high probability contains $\tilde{O}(n)$ edges. Hence, we can store each $G[V_i]$ in a single machine. This allows us to compute without any further communication for each edge $e$ in $G[V_i]$ the value $\tilde{x}_{e,T}$ and for each vertex $v$ in $G[V_i]$ the largest $t$ such that $v \in \tilde{V}^{active}_t$. Given this information, it is easy to compute in $O(1)$ MPC rounds $\tilde{x}_{e,T}$ for each edge $e$ in $E$ and therefore also the value $\tilde{x}_e$.
\end{proof}

\subsection{Putting Everything Together}
\label{sec:putting_evertyhing_together}

Finally, we can compute a $0.05$-tight solution by $O(\log \log \bar{d})$ invocations (in expectation) of \cref{alg:one_round_mpc}.

\begin{algorithm}[h]
\TitleOfAlgo{$FullMPC(G,b,r)$}
    \SetAlgoLined
    \KwData{$G=(V,E)$ is an unweighted graph, $b \in \mathbb{R}^V_{\geq 0}$, $r \in \mathbb{R}^{E}_{\geq 0}$}
    \KwResult{$x \in \mathbb{R}^E_{\geq 0}$}
    $E^{active} \leftarrow E$\;
    $\forall e \in E \colon x_e \leftarrow 0 $\;

    \While{$E^{active} \neq \emptyset$}{
        $G^{active} = (V,E^{active})$\;
        $\forall v \in V \colon b^{(rem)}_v \leftarrow b_v - \sum_{e \in E(v)} x_e$\;
        $\forall e \in E \colon r^{(rem)}_e \leftarrow r_e - x_e $\;
        \eIf{$|E^{active}| \geq n \log^{10}(n)$}{
            $x' \leftarrow  OneRoundMPC(G^{active}, b^{(rem)}, r^{(rem)}) $\;
        }{
            $x' \leftarrow Sequential(G^{active}, b^{(rem)}, r^{(rem)}, \lceil 100 \log(n)\rceil) $\; 
        }
        $\forall e \in E^{active}$: $x_e \leftarrow x_e + x'_e$\;
          $E^{active}\leftarrow E^{active} \cap E^{loose}(x,0.05)$
    }
    \Return{$x$}
\caption{ Complete Algorithm\label{alg:full_MPC}}
\end{algorithm}

\begin{lemma}
Let $x \leftarrow FullMPC(G,b,r)$. Then, $x$ is a feasible $0.05$-tight solution of the LP.
\end{lemma}
\begin{proof}
We verify by induction that $x$ is a feasible solution during the whole execution of \cref{alg:full_MPC}. At the beginning, $x$ is feasible. Now, during the while-loop, we set $x^{new}_e \leftarrow x^{old}_e + x'_e$ for every $e \in E^{active}$. As $0 \leq x'_e \leq r^{(rem)}_e \leq r_e - x^{old}_e$, it directly follows that $0 \leq x^{new}_e \leq r_e$. For every $v \in V$, we have

\[\sum_{e \in E(v)} x^{new}_e = \sum_{e \in E(v)} x^{old}_e + \sum_{e \in E(v) \cap E^{active}} x'_e \leq \sum_{e \in E(v)} x^{old}_e + b^{(rem)}_v = b_v.\]

At the end, $E^{active} = \emptyset$ and therefore also $E^{loose}(x,0.05) = \emptyset$. Hence, $x$ is indeed a $0.05$-tight solution.
\end{proof}

\begin{theorem}
Let $x \leftarrow FullMPC(G,b,r)$. Then, one can compute $x$ with positive constant probability in $O(\log\log \bar{d})$ MPC rounds using $\tilde{O}(n)$ local memory and $O(m + n)$ global memory.
\end{theorem}
\begin{proof}
We first show that the algorithm terminates after $O(\log\log \bar{d})$ iterations of the while-loop with positive constant probability.

Consider some fixed iteration and let $E^{active}_{old}$ denote the set of edges that are active at the beginning of the iteration and $E^{active}_{new}$ denote the set of edges that are active after the iteration. We say that the iteration is good if 

\[\frac{|E^{active}_{new}|}{n} \leq 2 \left(\frac{|E^{active}_{old}|}{n} \right)^{0.9999}.\]

Note that for every $e \in E^{active}_{new}$, it holds that $e \in E^{loose}(x',0.05)$ (with respect to $b^{(rem)}$ and $r^{(rem)}$). Therefore, it directly follows from $\cref{thm:expected_size}$ and $\cref{lem:loose_edges}$ together with a Markov bound that an iteration is good with probability at least $0.5$.
By another Markov bound, this implies that for a given $T$, with positive constant probability there are at least $\frac{T}{4}$ good iterations among the first $T$ iterations. As the average degree drops by a polynomial factor in each good iteration, a simple calculation gives that the average degree drops below $\log^{10}(n)$ after $O(\log \log \bar{d})$ good iterations, and therefore with positive constant probability also after $O(\log \log \bar{d})$ iterations. Once the average degree drops below $\log^{10}(n)$, it directly follows from $\cref{lem:loose_edges}$ that the algorithm needs at most one additional iteration to finish. 

It remains to discuss the MPC implementation. Note that for $x' \leftarrow OneRoundMPC(G^{active},b^{(rem)},r^{(rem)})$, $x'$ can be computed with high probability in $O(1)$ MPC rounds using $\tilde{O}(n)$ local memory and $O(m + n)$ global memory. Also, for $x' \leftarrow Sequential(G^{active}, b^{(rem)}, r^{(rem)}, \lceil 100 \log(n)\rceil) $, $x'$ can be computed in $O(1)$ MPC rounds as long as $G^{active}$ fits into a single machine.
Therefore, it follows by a union bound that the first $O(\log \log \bar{d})$ iterations can be simulated with high probability in $O(\log \log \bar{d})$ MPC rounds using $\tilde{O}(n)$ local memory and $O(m + n)$ global memory, which finishes the proof.
\end{proof}

\subsection{Proof of \cref{lem:vertex_bad}}
\label{sec:inductive_argument}

This section is dedicated to prove \cref{lem:vertex_bad}. The proof inductively shows that the idealized and the approximate process behave very similar by relating certain quantities in both processes.

We will repeatedly use the following Chernoff bound variant.

\begin{theorem}
\label{thm:Chernoff}
Suppose $X_1,X_2,\ldots,X_n$ are independent random variables taking values in $[0,1]$. Let $X := \sum_{i=1}^n X_i$. Then, for any $\delta \geq 0$,

\[Pr[X \geq (1+\delta)\E{X}] \leq e^{-\frac{1}{3}\min(\delta,\delta^2)\E{X}}\]

and for any $\delta \in [0,1]$,

\[Pr[X \leq (1-\delta) \E{X}] \leq e^{-\frac{1}{2}\delta^2\E{X}}.\]
\end{theorem}

Before proceeding with the inductive argument, we prove several lemmas which show that certain good events happen with high probability.
Most of these lemmas follow by a simple Chernoff Bound together with the fact that the values assigned to individual edges is sufficiently small throughout the process.
Some of these good events rely on that the random partitioning step behaves as expected and some oft them rely on that the random thresholds behave
as expected. 

\paragraph{Effect of The Random Partitioning}

We start with analysing the effect of the random partitioning. For the following two lemmas, we therefore assume that the random thresholds have been 
fixed in an arbitrary way in advance. After this fixing, the idealized process is a deterministic process and therefore the $x_{e,t}$ are fixed. As $\sum_{e \in E(v)} x_{e,t} \leq b_v$ and each edge $e \in E(v)$ is contained in $E^{local}(v)$ with probability $\frac{1}{N}$, we have $\E{\sum_{e \in E^{local}(v)} x_{e,t}} \leq \frac{b_v}{N}$. The lemma below shows concentration for $\sum_{e \in E^{local}(v)} x_{e,t}$.

\begin{lemma}
\label{lem:local_neighborhood_size}
    Assume that all the random thresholds are fixed in an arbitrary way.
    Consider an arbitrary $t \in \{0,1,\ldots,T\}$.
    Then, with high probability 
    \[\sum_{e \in E^{local}(v)} x_{e,t} \leq \frac{2b_v}{N}.\]  
\end{lemma}
\begin{proof}

For every $e \in E(v)$, let $X_e = x_{e,t} \cdot \frac{N^{1.99}}{b_v} \leq 1$  if $e \in E^{local}(v)$ and $X_e = 0$ otherwise. For $X = \sum_{e \in E(v)} X_e = \sum_{e \in E^{local}(v)} x_{e,t} \cdot \frac{N^{1.99}}{b_v}$, we have

\[\E{X} = \sum_{e \in E(v)} \E{X_e} = \sum_{e \in E(v)} x_{e,t} \cdot \frac{N^{1.99}}{b_v} \frac{1}{N} \leq N^{0.99}.\]

Note that $X$ is the sum of mutually independent random variables. Therefore, \cref{thm:Chernoff} implies that for $\delta = \frac{N^{0.99}}{\E{X}}$, we have

\[Pr[X \geq 2N^{0.99}] \leq Pr[X \geq (1+\delta)\E{X}] \leq e^{-\frac{1}{3} \min(\delta,\delta^2) \E{X}}  = e^{-\frac{N^{0.99}}{3}}.\]

Thus, with high probability $X \leq 2N^{0.99}$, which implies 

\[\sum_{e \in E^{local}(v)} x_{e,t} = X \cdot \frac{b_v}{N^{1.99}} \leq \frac{2 b_v}{N},\]

as needed.
\end{proof}

We have $\E{\sum_{e \in E^{local}(v) \cap E_t^{active}} x_{e,t-1}} = \frac{1}{N} \sum_{e \in E(v) \cap E_t^{active}} x_{e,t-1}$, as every edge $e \in E(v) \cap E_t^{active}$ is contained in $E^{local}(v) \cap E_t^{active}$. The next lemma shows concentration.

\begin{lemma}
\label{lem:sampling_concentration}
    Consider an arbitrary node $v \in V$ and $t \in [T]$.
    Assume we fix the randomness for choosing the random thresholds.
    Then, with high probability 
    \[|\sum_{e \in E(v) \cap E_t^{active}} x_{e,t-1} - N \cdot \sum_{e \in E^{local}(v) \cap E_t^{active}} x_{e,t-1}| \leq \frac{b_v}{N^{0.2}}\]
    
    and also
    
    \[|y_{v,0} - \tilde{y}_{v,0}| = |\sum_{e \in E(v) } x_{e,0} - N \cdot \sum_{e \in E^{local}(v)} \tilde{x}_{e,0}|=  |\sum_{e \in E(v) } x_{e,0} - N \cdot \sum_{e \in E^{local}(v)} x_{e,0}| \leq \frac{b_v}{N^{0.2}}.\]
\end{lemma}
\begin{proof}
For every $e \in E(v) \cap E^{active}_t$, let $X_e = x_{e,t-1} \cdot \frac{N^{1.99}}{b_v} \leq 1$  if $e \in E^{local}(v)$ and $X_e = 0$ otherwise. For $X = \sum_{e \in E(v) \cap E^{active}_t} X_e = N \cdot \sum_{e \in E^{local}(v) \cap E_t^{active}} \frac{N^{0.99}}{b_v}x_{e,t-1}$, we have

\[\E{X} = \sum_{e \in E(v) \cap E^{active}_t} \E{X_e} = \sum_{e \in E(v) \cap E^{active}_t} x_{e,t-1} \cdot \frac{N^{0.99}}{b_v} \leq N^{0.99} .\]

For $\delta = \frac{N^{0.7}}{\E{X}}$, \cref{thm:Chernoff} implies

\[Pr[X \geq N^{0.7} + \E{X}] = Pr[X \geq (1+\delta)\E{X}] \leq e^{-\frac{1}{3}\min(\delta,\delta^2) \E{X}} \leq e^{-\frac{1}{3}\min(N^{0.7},\frac{N^{1.4}}{\E{X}})} \leq e^{-N^{0.4}}.\]

If $\E{X} \leq N^{0.7}$, then the above bound already implies that with high probability $|X-\E{X}| \leq N^{0.7}$.
It remains to consider the case $\E{X}\geq N^{0.7}$. Then, for $\delta = \frac{N^{0.7}}{\E{X}}$, we obtain

\[Pr[X \leq -N^{0.7} + \E{X}] = Pr[X \leq (1 - \delta)\E{X}] \leq e^{-\frac{1}{2}\delta^2\E{X}} \leq e^{-N^{0.4}}.\]

Thus, combining the two bounds also implies in this case that $|X-\E{X}| \leq N^{0.7}$ with high probability.
Moreover, if $|X-\E{X}| \leq N^{0.7}$, then

\[|\sum_{e \in E(v) \cap E_t^{active}} x_{e,t-1} - N \cdot \sum_{e \in E^{local}(v) \cap E_t^{active}} x_{e,t-1}| = \frac{b_v}{N^{0.99}}|X - \E{X}| \leq \frac{b_v}{N^{0.2}},\]
as needed. We omit the proof that $|y_{v,0} - \tilde{y}_{v,0}| \leq \frac{b_v}{N^{0.2}}$ as it follows along the exact same line.
\end{proof}

\paragraph{Effect of The Random Thresholds}

Next, we analyze the effect of the random partitioning. In particular, we take the point of view that we have completely executed the approximate process up to some iteration $t$ and now we select the random thresholds for iteration $t$, which then allows us to simulate iteration $t$.

The first lemma states that if a vertex is still active after iteration $t-1$ and $y_{v,t-1}$ and its approximation $\tilde{y}_{v,t-1}$ are sufficiently close to each other, then the probability that $v$ behaves differently in iteration $t$ in the two processes in iteration $t$, i.e., it is active in the idealized process but not in the approximate process in iteration $t$, or vice versa, is small. The lemma directly follows from the algorithm description of the idealized and randomized process.
\begin{lemma}
\label{lem:simple_threshold}
Consider an arbitrary iteration $t \in [T]$. We assume that all the randomness is fixed except for the randomness used for the thresholds in iteration $t$.  Let $v \in V^{active}_{t-1} \cap \tilde{V}^{active}_{t-1}$ and $\sigma \in \mathbb{R}_{\geq 0}$ such that $|y_{v,t-1} - \tilde{y}_{v,t-1}| \leq \sigma b_v$. Then, 

\[Pr[v \in V_t^{active} \triangle \tilde{V}_t^{active}] \leq 5 \sigma.\]
\end{lemma}
\begin{proof}
We have 

\[Pr[v \in V_t^{active} \triangle \tilde{V}_t^{active}] = Pr[\mathcal{T}_{v,t} \in [\min(y_{v,t-1},\tilde{y}_{v,t-1}), \max(y_{v,t-1},\tilde{y}_{v,t-1}) ]] \leq \frac{\sigma b_v}{b_v/5} = 5\sigma.\]
\end{proof}

The next lemma uses the previous lemma together with a simple concentration argument to bound, for a given vertex $v$, the effect of "bad" random thresholds of neighbors of $v$.

\begin{lemma}
\label{lem:thresholds}
Consider an arbitrary iteration $t \in [T]$. We assume that all the randomness is fixed except for the randomness used for the thresholds in iteration $t$.
Let $v \in V$ be arbitrary. Let $F \subseteq E(v) \cap E^{active}_{t-1} \cap \tilde{E}^{active}_{t-1}$ and let $\sigma$ such that for every $u \in V^{active}_{t-1} \cap \tilde{V}^{active}_{t-1}$, $|y_{u,t-1} - \tilde{y}_{u,t-1}| \leq \sigma b_u$. 
Then, with high probability $v \in V^{active}_t \triangle \tilde{V}^{active}_t$ or

\[\sum_{e \in F} |x_{e,t} - \tilde{x}_{e,t}| \leq \max \left( 10 \sigma \sum_{e \in F} x_{e,t-1}, \frac{b_v}{N^{1.3}} \right).\]
\end{lemma}
\begin{proof}
Assume we additionally fix the randomness for the threshold $\mathcal{T}_{v,t}$. We assume that $v \in V^{active}_t \cap \tilde{V}^{active}_t$, as otherwise $v \in V_t^{active} \triangle \tilde{V}_t^{active}$ or  $\sum_{e \in F} |x_{e,t} - \tilde{x}_{e,t}| = 0$.
For every $e \in F$, let $X_e = x_{e,t-1} \cdot \frac{N^{1.99}}{b_v} \leq 1$  if $e \in E^{active}_t \triangle \tilde{E}^{active}_t$ and $X_e = 0$ otherwise. For 
\[ X := \sum_{e \in F} X_e = \sum_{e \in F \cap (E^{active}_t \triangle \tilde{E}^{active}_t)} x_{e,t-1} \cdot \frac{N^{1.99}}{b_v} = \sum_{e \in F \cap (E^{active}_t \triangle \tilde{E}^{active}_t)} |x_{e,t} - \tilde{x}_{e,t}| \cdot \frac{N^{1.99}}{b_v} = \sum_{e \in F} |x_{e,t} - \tilde{x}_{e,t}| \cdot \frac{N^{1.99}}{b_v},\]

we have

\[\E{X} = \sum_{e \in F} \E{X_e} \leq \sum_{e \in F} 5\sigma x_{e,t-1} \cdot \frac{N^{1.99}}{b_v} .\]

Note that $X$ is the sum of mutually independent random variables. Thus, for $\delta = \max \left(1,\frac{N^{0.6}}{\E{X}}\right)$, we have

\[Pr \left[ X \geq \frac{N^{1.99}}{b_v} \max \left( 10 \sigma \sum_{e \in F} x_{e,t-1}, \frac{b_v}{N^{1.3}} \right) \right]  \leq Pr[X \geq (1+\delta)\E{X}] \leq e^{-\frac{1}{3}\min(\delta,\delta^2)\E{X}} \leq e^{-\frac{1}{3}N^{0.6}}.\]

Thus, with high probability $X \leq \frac{N^{1.99}}{b_v} \max \left( 10 \sigma \sum_{e \in F} x_{e,t-1}, \frac{b_v}{N^{1.3}} \right)$, which directly implies 

\[\sum_{e \in F} |x_{e,t} - \tilde{x}_{e,t}| \leq \max \left( 10 \sigma \sum_{e \in F} x_{e,t-1}, \frac{b_v}{N^{1.3}} \right) ,\]

as needed.

\end{proof}

Next, we present two simple corollaries of the previous lemma.

\begin{corollary}
\label{cor:thresh_easy}
Consider an arbitrary iteration $t \in [T]$. We assume that all the randomness is fixed except for the randomness used for the thresholds in iteration $t$. Let $v \in V$ be arbitrary and let $\sigma \geq \frac{1}{N^{0.2}}$ such that for every $u \in V^{active}_{t-1} \cap \tilde{V}^{active}_{t-1}$ $|y_{u,t-1} - \tilde{y}_{u,t-1}| \leq \sigma b_u$. Then, with high probability $v \in V^{active}_t \Delta \tilde{V}^{active}_t$ or

\[\sum_{e \in E(v) \cap E^{active}_{t-1} \cap \tilde{E}^{active}_{t-1}} |x_{e,t} - \tilde{x}_{e,t}| \leq 10\sigma b_v.\]
\end{corollary}
\begin{proof}
Directly follows from \cref{lem:thresholds} by setting $F = E(v) \cap E^{active}_{t-1} \cap \tilde{E}^{active}_{t-1}$ together with the fact that

\[\max \left( 10 \sigma \sum_{e \in E(v) \cap E^{active}_{t-1} \cap \tilde{E}^{active}_{t-1}} x_{e,t-1}, \frac{b_v}{N^{1.3}} \right) \leq \max \left( 10\sigma b_v,\frac{b_v}{N^{1.3}} \right) = 10\sigma b_v,\]
where we used $\sigma \geq \frac{1}{N^{0.2}}$.
\end{proof}
\begin{corollary}
\label{cor:thresh_hard}
Consider an arbitrary iteration $t \in [T]$. We assume that all the randomness is fixed except for the randomness used for the thresholds in iteration $t$. Let $v \in V$ be arbitrary. Assume that $\sum_{e \in E^{local}(v)} x_{e,t-1} \leq \frac{2b_v}{N}$ and let $\sigma \geq \frac{1}{N^{0.2}}$ such that for every $u \in V^{active}_{t-1} \cap \tilde{V}^{active}_{t-1}$ $|y_{u,t-1} - \tilde{y}_{u,t-1}| \leq \sigma b_u$. Then, with high probability $v \in V^{active}_t \triangle \tilde{V}^{active}_t$ or 

\[\sum_{e \in E^{local}(v) \cap E^{active}_{t-1} \cap \tilde{E}^{active}_{t-1}} |x_{e,t} - \tilde{x}_{e,t}| \leq 20\sigma \frac{b_v}{N}.\]
\end{corollary}

\begin{proof}
Directly follows from \cref{lem:thresholds} by setting $F = E^{local}(v) \cap E^{active}_{t-1} \cap \tilde{E}^{active}_{t-1}$ together with the fact that

\[\max \left( 10 \sigma \sum_{e \in E^{local}(v) \cap E^{active}_{t-1} \cap \tilde{E}^{active}_{t-1}} x_{e,t-1}, \frac{b_v}{N^{1.2}} \right) \leq \max \left(10\sigma \frac{2b_v}{N},\frac{b_v}{N^{1.3}} \right) = \frac{20\sigma b_v}{N},\]
where we used $\sigma \geq \frac{1}{N^{0.2}}$.
\end{proof}

\paragraph{Overall Effect of the Randomness}

The following theorem can be seen as a summary of all the previous lemmas that analysed the effect of the random partitioning and the random thresholds.
\begin{theorem}
\label{thm:random_main}
The following holds with high probability.
First, $|y_{v,0} - \tilde{y}_{v,0}| \leq \frac{b_v}{N^{0.2}}$.
Second, for every $t \in [T]$, let \[\sigma_{t-1} := \max \left( \frac{1}{N^{0.2}},\max_{u \in V_{t-1}^{active} \cap \tilde{V}^{active}_{t-1}} \frac{|y_{u,t-1} - \tilde{y}_{u,t-1}|}{b_u} \right).\]

Then, for every $v \in V^{active}_t \cap \tilde{V}^{active}_t$ it holds that

\begin{enumerate}
        \item $ |\sum_{e \in E(v) \cap E_t^{active}} x_{e,t-1} - N \cdot \sum_{e \in E^{local}(v) \cap E_t^{active}} x_{e,t-1}| \leq \frac{b_v}{N^{0.2}}.$
        \item $\sum_{e \in E(v) \cap E^{active}_{t-1} \cap \tilde{E}^{active}_{t-1}} |x_{e,t} - \tilde{x}_{e,t}| \leq 10\sigma_{t-1} b_v.$ 
        \item $\sum_{e \in E^{local}(v) \cap E^{active}_{t-1} \cap \tilde{E}^{active}_{t-1}} |x_{e,t} - \tilde{x}_{e,t}| \leq 20\sigma_{t-1} \frac{b_v}{N}.$
\end{enumerate}
\end{theorem}

\begin{proof}
Both $|y_{v,0} - \tilde{y}_{v,0}| \leq \frac{b_v}{N^{0.1}}$ and the first condition hold with high probability according to \cref{lem:sampling_concentration} together with a union bound over all choices of $t \in [T]$ and $v \in V$.
The second condition holds with high probability according to \cref{cor:thresh_easy}, using the fact that for every $u \in V_{t-1}^{active} \cap \tilde{V}^{active}_{t-1}$, $|y_{u,t-1} - \tilde{y}_{u,t-1}| \leq \sigma_{t-1}b_u$, together with a union bound over all choices of $t \in [T]$ and $v \in V$.
For the third condition, consider a fixed $t \in [T]$ and $v \in V$. 
First, fix the randomness for the thresholds up to iteration $t-1$ in an arbitrary way. 
Now, reveal the randomness for the random partitioning. 
According to \cref{lem:local_neighborhood_size}, $\sum_{e \in E^{local}(v)} x_{e,t-1} \leq \frac{2b_v}{N}$. 
If after the revealing the randomness for the partitioning $v \notin V^{active}_{t-1} \cap \tilde{V}^{active}_{t-1}$, then there is nothing to show. 
Otherwise, as for every $u \in V_{t-1}^{active} \cap \tilde{V}^{active}_{t-1}$, $|y_{u,t-1} - \tilde{y}_{u,t-1}| \leq \sigma_{t-1}b_u$, \cref{cor:thresh_hard} implies that $v \in V_t^{active} \triangle \tilde{V}^{active}_t$ or $\sum_{e \in E^{local}(v) \cap E^{active}_{t-1} \cap \tilde{E}^{active}_{t-1}} |x_{e,t} - \tilde{x}_{e,t}| \leq 20\sigma_{t-1} \frac{b_v}{N}$. 
Hence, if $v \in V^{active}_t \cap \tilde{V}^{active}_t$ the latter condition must hold, as desired. The theorem follows by a union bound.

\end{proof}

\paragraph{Induction}

We are now ready for the induction. We start with a very simple lemma before we present the main induction.

\begin{lemma}
\label{lem:weight_change}
 Consider an arbitrary $t \in [T]$ and $e \in E^{active}_{t-1} \Delta \tilde{E}^{active}_{t-1}$. Then, 
 
 \[ |x_{e,t} - \tilde{x}_{e,t}| \leq 3 \cdot |x_{e,t-1} - \tilde{x}_{e,t-1}|.\]
\end{lemma}
\begin{proof}

We only consider the case that $e \in E^{active}_{t-1} \setminus \tilde{E}^{active}_{t-1}$, which is easily seen to imply $x_{e,t-1} \geq 2\tilde{x}_{e,t-1}$. The case $e \in \tilde{E}^{active}_{t-1} \setminus E^{active}_{t-1}$ is completely symmetrical. We have

\begin{align*}
|x_{e,t} - \tilde{x}_{e,t}| &= x_{e,t} - \tilde{x}_{e,t} \\
&\leq 2x_{e,t-1} - \tilde{x}_{e,t-1} \\
&\leq x_{e,t-1} + (x_{e,t-1} - \tilde{x}_{e,t-1}) \\
&\leq 2(x_{e,t-1} - \tilde{x}_{e,t-1}) + (x_{e,t-1} - \tilde{x}_{e,t-1}) \\
&\leq 3 \cdot |x_{e,t-1} - \tilde{x}_{e,t-1}|,
\end{align*}
as needed.
\end{proof}

We are now ready for the main theorem. In order to prove $\cref{lem:vertex_bad}$, it would suffice to show that $\tilde{y}_{v,t}$ approximates $y_{v,t}$ sufficiently well most of the times. However, for the induction to work, we not only need to track $|y_{v,t} - \tilde{y}_{v,t}|$ but two additional quantities as well, namely $\sum_{e \in E(v)} |x_{e,t} - \tilde{x}_{e,t}|$ and $\sum_{e \in E^{local}(v)} |x_{e,t} - \tilde{x}_{e,t}|$.

\begin{theorem} [Evolution of Weight Estimates]
\label{thm:evolution_of_weight_estimates}
Let $\rho_t := N^{-0.2} \cdot 100^t$. The following holds with high probability for every $t \in \{0,1,\ldots,T\}$ and $v \in V_t^{active} \cap \tilde{V}_t^{active}$:

\begin{enumerate}
     \item $|y_{v,t} - \tilde{y}_{v,t}| \leq \rho_t b_v$
    \item $\sum_{e \in E(v)} |x_{e,t} - \tilde{x}_{e,t}| \leq \rho_t b_v$ 
    \item $\sum_{e \in E^{local}(v)} |x_{e,t} - \tilde{x}_{e,t}| \leq \frac{\rho_t b_v}{N}.$
\end{enumerate}

\end{theorem}
\begin{proof}

Assume that all the conditions stated in \cref{thm:random_main} are satisfied (which happens with high probability).
We show by induction on $t$ that this implies that the conditions stated in the theorem statement are satisfied.

We start with the base case $t = 0$. The first condition of the base case is one of the conditions of \cref{thm:random_main}. The second and third condition of the base case trivially hold as $x_{e,0} = \tilde{x}_{e,0}$ for every $e \in E$. 

For the induction step, consider an arbitrary $t \in [T]$ and assume that for every $v \in V_{t-1}^{active} \cap \tilde{V}_{t-1}^{active}$ 

\begin{enumerate}
    \item $|y_{v,t-1} - \tilde{y}_{v,t-1}| \leq \rho_{t-1} b_v$,
    \item $\sum_{e \in E(v)} |x_{e,t-1} - \tilde{x}_{e,t-1}| \leq \rho_{t-1} b_v$ and 
    \item $\sum_{e \in E^{local}(v)} |x_{e,t-1} - \tilde{x}_{e,t-1}| \leq \frac{\rho_{t-1} b_v}{N}.$
\end{enumerate}

The first condition together with the fact that $\rho_{t-1} \geq N^{-0.2}$ implies that 

 \[\rho_{t-1} \geq \max \left( \frac{1}{N^{0.2}}, \max_{u \in V_{t-1}^{active} \cap \tilde{V}^{active}_{t-1}} \frac{|y_{u,t-1} - \tilde{y}_{u,t-1}|}{b_u} \right).\]
 
 Hence, as we assume that the conditions stated in \cref{thm:random_main} are satisfied, for every $v \in V_{t}^{active} \cap \tilde{V}_{t}^{active}$ we have
 
 \begin{enumerate}
             \item $ |\sum_{e \in E(v) \cap E_t^{active}} x_{e,t-1} - N \cdot \sum_{e \in E^{local}(v) \cap E_t^{active}} x_{e,t-1}| \leq  \rho_{t-1} b_v.$
        \item $\sum_{e \in E(v) \cap E^{active}_{t-1} \cap \tilde{E}^{active}_{t-1}} |x_{e,t} - \tilde{x}_{e,t}| \leq 10 \rho_{t-1} b_v.$ 
        \item $\sum_{e \in E^{local}(v) \cap E^{active}_{t-1} \cap \tilde{E}^{active}_{t-1}} |x_{e,t} - \tilde{x}_{e,t}| \leq 20 \rho_{t-1} \frac{b_v}{N}.$
 \end{enumerate}
 
Consider an arbitrary $v \in V^{active}_t \cap \tilde{V}^{active}_t$.
We start by showing that $\sum_{e \in E(v)} |x_{e,t} - \tilde{x}_{e,t}| \leq \rho_t b_v$.

We have

\begin{align*}
     \sum_{e \in E(v)} |x_{e,t} - \tilde{x}_{e,t}| &=  \sum_{e \in E(v) \cap( E_{t-1}^{active} \cap \tilde{E}_{t-1}^{active})} |x_{e,t} - \tilde{x}_{e,t}| \\
     &+ \sum_{e \in E(v) \cap( E_{t-1}^{active} \triangle \tilde{E}_{t-1}^{active})} |x_{e,t} - \tilde{x}_{e,t}| \\
     &+ \sum_{e \in E(v) \setminus (E_{t-1}^{active} \cup \tilde{E}_{t-1}^{active})} |x_{e,t} - \tilde{x}_{e,t}|   \\
     &\leq 10 \rho_{t-1} b_v + 3\sum_{e \in E(v)} |x_{e,t-1} - \tilde{x}_{e,t-1}| + \sum_{e \in E(v)} |x_{e,t-1} - \tilde{x}_{e,t-1}| \\
     &\leq (10 + 3 + 1)\rho_{t-1}b_v \\
     &\leq \rho_t b_v.
\end{align*}

Along the exact same line, it follows that

\begin{align*}
   \sum_{e \in E^{local}(v)} |x_{e,t} - \tilde{x}_{e,t}| &=  \sum_{e \in E^{local}(v) \cap( E_{t-1}^{active} \cap \tilde{E}_{t-1}^{active})} |x_{e,t} - \tilde{x}_{e,t}| \\
     &+ \sum_{e \in E^{local}(v) \cap( E_{t-1}^{active} \triangle \tilde{E}_{t-1}^{active})} |x_{e,t} - \tilde{x}_{e,t}| \\
     &+ \sum_{e \in E^{local}(v) \setminus (E_{t-1}^{active} \cup \tilde{E}_{t-1}^{active})} |x_{e,t} - \tilde{x}_{e,t}|   \\
     &\leq \frac{20 \rho_{t-1} b_v}{N} + 3\sum_{e \in E^{local}(v)} |x_{e,t-1} - \tilde{x}_{e,t-1}| + \sum_{e \in E^{local}(v)} |x_{e,t-1} - \tilde{x}_{e,t-1}| \\
     &\leq \frac{(20 + 3 + 1) \rho_{t-1} b_v}{N} \\
     &\leq \frac{\rho_t b_v}{N}. \\
\end{align*}

It remains to show that $|y_{v,t} - \tilde{y}_{v,t}| \leq \rho_{t} b_v$. We have

\begin{align*}
    |y_{v,t} - \tilde{y}_{v,t}| &= |y_{v,t} - y_{v,t-1} - ( \tilde{y}_{v,t} - \tilde{y}_{v,t-1}) + y_{v,t-1} - \tilde{y}_{v,t-1}| \\
    &\leq | \sum_{e \in E(v) \cap E^{active}_t} x_{e,t-1} - N \cdot \sum_{e \in E^{local}(v) \cap \tilde{E}^{active}_t} \tilde{x}_{e,t-1}| + |y_{v,t-1} - \tilde{y}_{v,t-1}| 
\end{align*}

As $|y_{v,t-1} - \tilde{y}_{v,t-1}| \leq \rho_{t-1} b_v$, we can focus on finding an upper bound for the first term. We have

\begin{align*}
    &| \sum_{e \in E(v) \cap E^{active}_t} x_{e,t-1} - N \cdot \sum_{e \in E^{local}(v) \cap \tilde{E}^{active}_t} \tilde{x}_{e,t-1}| \\
    \leq & | \sum_{e \in E(v) \cap E^{active}_t} x_{e,t-1} - N \cdot \sum_{e \in E^{local}(v) \cap E^{active}_t} \tilde{x}_{e,t-1}| 
    + N \cdot \sum_{e \in E^{local}(v) \cap (E^{active}_t \triangle \tilde{E}^{active}_t)} \tilde{x}_{e,t-1}.
\end{align*}

We upper bound the two terms separately. For the first term, we have

\begin{align*}
    &| \sum_{e \in E(v) \cap E^{active}_t} x_{e,t-1} - N \cdot \sum_{e \in E^{local}(v) \cap E^{active}_t} \tilde{x}_{e,t-1}| \\
    \leq & | \sum_{e \in E(v) \cap E^{active}_t} x_{e,t-1} - N \cdot \sum_{e \in E^{local}(v) \cap E^{active}_t} x_{e,t-1}| + N \cdot \sum_{e \in E^{local}(v)} |x_{e,t-1} - \tilde{x}_{e,t-1}| \\
    \leq & 2 \rho_{t-1} b_v.
\end{align*}

For the second term, ignoring the $N$-factor, we have

\[\sum_{e \in E^{local}(v) \cap (E^{active}_t \triangle \tilde{E}^{active}_t)} \tilde{x}_{e,t-1} \leq \sum_{e \in E^{local}(v)} |x_{e,t} - \tilde{x}_{e,t}| \leq \frac{24 \rho_{t-1} b_v}{N}.\]

Hence, adding up all the contributions, we obtain

\[|y_{v,t} - \tilde{y}_{v,t}| \leq (1 + 2 + 24)\rho_{t-1}b_v \leq \rho_t b_v,\]

which finishes the induction proof.

\end{proof}

The previous theorem has shown that $|y_{v,t}- \tilde{y}_{v,t}|$ is small as long as $v$ is active in iteration $t$ both in the 
idealized and the approximate process. Because of the random thresholds, this suffices to show that it is unlikely for a vertex that
there exists an iteration $t$ such that $t$ is active in iteration $t$ in exactly one of the two processes, i.e., that $v \in V_t^{active} \triangle \tilde{V}_t^{active}$.

\begin{theorem}
\label{thm:bad_vertex}
Let $v \in V$ be arbitrary. Then, for every $t \in \{0,1,\ldots,T\}$, we have 
\[Pr[v \in V_t^{active} \triangle \tilde{V}_t^{active}] \leq \frac{t}{n} + \sum_{\ell=0}^t \rho_\ell.\] 
\end{theorem}

\begin{proof}
We prove the statement by induction on $t$. The base case $t = 0$ trivially holds as $V_0^{active} = \tilde{V}_0^{active}$.

Now, consider an arbitrary $t \in [T]$ and assume that the statement holds for $t-1$. We have

\[Pr[v \in V^{active}_t \triangle \tilde{V}^{active}_t] \leq Pr[v \in V^{active}_{t-1} \triangle \tilde{V}^{active}_{t-1}] +Pr[v \in V^{active}_t \triangle \tilde{V}^{active}_t | v \in V^{active}_{t-1} \cap \tilde{V}^{active}_{t-1}]\]

By induction, we have $Pr[v \in V^{active}_{t-1} \triangle \tilde{V}^{active}_{t-1}] \leq \frac{t-1}{n} + \sum_{\ell = 0}^{t-1} \rho_\ell$.
Hence, it remains to show that $Pr[v \in V^{active}_t \triangle \tilde{V}^{active}_t | v \in V^{active}_{t-1} \cap \tilde{V}^{active}_{t-1}] \leq \frac{1}{n} + \rho_t$.

Let $A_{t-1}$ be the event that for every $u \in V^{active}_{t-1} \cap \tilde{V}^{active}_{t-1}$, it holds that $|y_{u,t-1} - \tilde{y}_{u,t-1}| \leq \rho_{t-1}$. According to \cref{thm:evolution_of_weight_estimates}, we have $Pr[A_{t-1}] \geq 1 - \frac{1}{n}$. Therefore,

\begin{align*}
    Pr[v \in V^{active}_t \triangle \tilde{V}^{active}_t | v \in V^{active}_{t-1} \cap \tilde{V}^{active}_{t-1}] &\leq Pr[v \in V^{active}_t \triangle \tilde{V}^{active}_t | (v \in V^{active}_{t-1} \cap \tilde{V}^{active}_{t-1}) \cap A_{t-1}] + Pr[\overline{A_{t-1}}] \\
    &\leq  Pr[v \in V^{active}_t \triangle \tilde{V}^{active}_t | (v \in V^{active}_{t-1} \cap \tilde{V}^{active}_{t-1}) \cap A_{t-1}] + \frac{1}{n}.
\end{align*}

Thus, it remains to show that $Pr[v \in V^{active}_t \triangle \tilde{V}^{active}_t | (v \in V^{active}_{t-1} \cap \tilde{V}^{active}_{t-1}) \cap A_{t-1}]\leq \rho_t$.
As $v \in V^{active}_{t-1} \cap \tilde{V}^{active}_{t-1}$,
it follows from the definition of $A_{t-1}$ that $|y_{v,t-1} - \tilde{y}_{v,t-1}| \leq \rho_{t-1}b_v$ and therefore \cref{lem:simple_threshold} directly gives that 

\[Pr[v \in V^{active}_t \triangle \tilde{V}^{active}_t | (v \in V^{active}_{t-1} \cap \tilde{V}^{active}_{t-1}) \cap A_{t-1}] \leq 5\rho_{t-1} \leq \rho_t,\]

which finishes the proof.

\end{proof}

Finally, we are ready to prove \cref{lem:vertex_bad}.

\begin{proof}[Proof of \cref{lem:vertex_bad}]
Let $A$ be the event that $v \notin \bigcup_{t=0}^T (V_t^{active} \triangle \tilde{V}_t^{active})$ and $B$ the event that the "with high probability" guarantee of \cref{thm:evolution_of_weight_estimates} holds.
We first show that $Pr[A \cap B] \geq 1 - \frac{1}{N^{0.1}}$ and afterwards that $A \cap B$ implies $|\sum_{e \in E(v)} x_{e,T} - \sum_{e \in E(v)} \tilde{x}_{e,T}| \leq 0.1b_v$.
According to \cref{thm:bad_vertex},

\[Pr[v \in V_t^{active} \triangle \tilde{V}_t^{active}] \leq \frac{t}{n} + \sum_{\ell = 0}^t \rho_\ell 
\leq 2 \cdot \rho_T = 2 \cdot \frac{1}{N^{0.2}} \cdot 100^{\lfloor \log_2(N)/1000 \rfloor} \leq \frac{1}{N^{0.15}}.\]
Therefore,

\[Pr[A \cap B] \geq 1 - \frac{(T+1)}{N^{0.15}} - \frac{1}{n} \geq 1 - \frac{1}{N^{0.1}}.\]

Now, assume that $A \cap B$ holds.
Consider the largest $t \in \{0,1,\ldots,T\}$ with $v \in V^{active}_t \cap \tilde{V}^{active}_t$. As $B$ holds, we have

\[|\sum_{e \in E(v)} x_{e,t} - \sum_{e \in E(v)} \tilde{x}_{e,t}| \leq \sum_{e \in E(v)} |x_{e,t} - \tilde{x}_{e,t}| \leq 0.1b_v.\]

By definition of $t$, for every $t' > t$, $v \notin V^{active}_{t'} \cap \tilde{V}^{active}_{t'}$ and as $v \notin V^{active}_{t'} \triangle \tilde{V}^{active}_{t'}$ it follows that $v \notin V^{active}_{t'} \cup \tilde{V}^{active}_{t'}$. In particular, this implies for every $e \in E(v)$ that $x_{e,t'} = x_{e,t'-1}$ and $\tilde{x}_{e,t'} = \tilde{x}_{e,t'-1}$. Hence, it follows from a simple induction that

\[|\sum_{e \in E(v)} x_{e,T} - \sum_{e \in E(v)} \tilde{x}_{e,T}| = |\sum_{e \in E(v)} x_{e,t} - \sum_{e \in E(v)} \tilde{x}_{e,t}|  \leq 0.1b_v,\]

as desired.
\end{proof}

\subsection{Proof of \cref{lem:local_memory}}

\begin{lemma}
\label{lem:local_memory}

For every $i \in [N]$, the graph $G[V_i]$ contains $\tilde{O}(n)$ edges with high probability.

\end{lemma}
\begin{proof}
Consider an arbitrary $i \in [N]$. For every $v \in V$, we set $X_v = \frac{|E(v)|}{n} \leq 1$ if $v \in V_i$ and $X_v = 0$ otherwise. 
Let $X = \sum_{v \in V} X_v = \sum_{v \in V_i} \frac{|E(v)|}{n}$. We have

\[\E{X} = \sum_{v \in V} \E{X_v} = \sum_{v \in V} \frac{1}{N}\frac{|E(v)|}{n} = \frac{1}{N} \frac{2|E|}{n} \leq N \cdot \frac{2|E|}{\bar{d}n} = N.\]

By Chernoff,

\[Pr[X \geq 2N] \leq e^{-\frac{1}{3}N}.\]

Hence, with high probability

\[\sum_{v \in V_i} |E(v)| \leq N \cdot n.\]

Now, consider an arbitrary $v \in V$. For every $e \in E(v)$, let $Y_e = 1$ if $e \in E^{local}(v)$ and $Y_e = 0$ otherwise. We define $Y = \sum_{e \in E(v)} Y_e = |E^{local}(v)|$. We have $\E{Y} = \frac{|E(v)}{N}$.
Setting $\delta = \max(1,\frac{\log^2(n)}{\E{Y}})$, a Chernoff bound implies

\[Pr \left[ Y \geq \max(1 + \log^2(n),2\frac{|E(v)|}{N}) \right] = Pr[Y \geq (1+\delta) \E{Y}] \leq e^{-\frac{1}{3}\log^2(n)}.\]

Therefore, with high probability

\[\sum_{v \in V_i} |E^{local}(v)| \leq \sum_{v \in V_i}(1 + \log^2(n) + 2\frac{|E(v)|}{N}) = \tilde{O}(n),\]

as needed.
\end{proof}

\subsection{Proof of \cref{lem:matching_from_tight_solution}}
\label{sec:matching_from_tight}

\begin{proof}
Let $OPT$ denote the optimum value of the LP. We first show that $\sum_{e \in E} x_e \geq \frac{\alpha}{3}$ by using duality. The dual of the LP is

	\[
	\begin{array}{lcll}
   &\text{minimize} &\sum_{v \in V} b_v y_v +\sum_{e \in E} z_e  \\
   &\text{subject to} &  y_u + y_v + z_e \geq 1 &\text{for every $e = \{u,v\} \in E$} \\ 

   & &y,z \ge 0.
\end{array} \]

For every $v \in V$, we set $y_v = 1$ if $\sum_{e \in E(v)} x_e \geq \alpha b_v$ and otherwise $y_v = 0$. For every $e \in E$, we set $z_e = 1$ if $x_e \geq \alpha$ and $z_e = 0$ otherwise. It directly follows from the $\alpha$-tightness of $x$ that $(y,z)$ is dual feasible. Moreover, a simple charging argument gives $\sum_{e \in E} x_e \geq \frac{\alpha}{3} \left(\sum_{v \in V} b_v y_v + \sum_{e \in E} z_e \right)$. Hence, it follows from weak duality that $\sum_{e \in E} x_e \geq \frac{\alpha}{3}$. As the primal LP is a relaxation of the $b$-matching problem, it therefore suffices to show how to compute a $b$-matching $M$ such that with positive constant probability $|M| \geq \frac{1}{100}\sum_{e \in E(v)} x_e$. Let $M$ be the $b$-matching one obtains by first sampling each edge $e$ independently with probability $\frac{x_e}{4}$ and then including each sampled edge $\{u,v\}$ if the number of sampled edges incident to $u$ is at most $b_u$ and the number of sampled edges incident to $v$ is at most $b_v$.
Let $e = \{u,v\}$ be an arbitrary edge with $x_e > 0$. 
We denote by $A_e$ the event that $e$ gets sampled and by $A_u$ ($A_v$) the event that the number of sampled edges incident to $u$ ($v$) is at most $b_u$ ($b_v$).
Note that

\[Pr[A_e \cap A_u \cap A_v] = Pr[A_u \cap A_v | A_e] \cdot Pr[A_e] = Pr[A_u|A_e] \cdot Pr[A_v|A_e] \cdot \frac{x_e}{4}.\]

 Note that the expected number of edges incident to $b_u$ conditioned on $A_e$ is at most $1 + \frac{b_u}{4} \leq \frac{3b_u}{4}$. Hence, by a Markov Bound it follows that the probability that there are at least $b_u + 1$ sampled edges incident to $u$ is at most $\frac{1 + \frac{b_u}{4}}{1 + b_u} \leq \max \left( \frac{1 + \frac{1}{4}}{1 + 1},\frac{\frac{b_u}{2} + \frac{b_u}{4}}{b_u} \right) \leq \frac{3}{4}$. Therefore,

\[Pr[A_e \cap A_u \cap A_v] \geq \frac{1}{4} \cdot \frac{1}{4} \cdot \frac{x_e}{4}.\]

Hence, the expected size of $M$ is at least $\frac{1}{64}\sum_{e \in E(v)} x_e$ and therefore with positive constant probability the size of the matching is at least $\frac{1}{100}\sum_{e \in E(v)} x_e$. The sampling procedure can easily be implemented in MPC and therefore this finishes the proof of the lemma.
\end{proof}

\begin{proof}[Proof of \cref{lem:primal_feasible}]
We prove the statement by induction on $t$. First, we consider the base case $t = 0$. Clearly, $0 \leq x_{e,t} \leq r_e$ for every $e \in E$. Now, consider an arbitrary $v \in V$. Clearly, $\sum_{e \in E(v)} x_{e,0} \geq 0$ and moreover

\[\sum_{e \in E(v)} x_{e,0} \leq |E(v)| \cdot 0.8\frac{b_v}{|E(v)|} \leq 0.8b_v.\]
Now, let $t \in [T]$ be arbitrary and assume the statement holds for $t-1$. First, consider an arbitrary $e \in E$. Clearly, $x_{e,t} \geq 0$. If $x_{e,t-1} > r_e/2$, then $x_{e,t} = x_{e,t-1} \leq r_e$. If $x_{e,t-1} \leq r_e/2$, then $x_{e,t} \leq 2x_{e,t-1} \leq r_e$, as desired. Next, consider an arbitrary $v \in V$. Clearly, $\sum_{e \in E(v)} x_{e,t} \geq 0$. If $\sum_{e \in E(v)} x_{e,t-1} > 0.4 b_v$, then $v \notin V^{active}_t$ and therefore $E(v) \cap E^{active}_t = \emptyset$. This implies $\sum_{e \in E(v)} x_{e,t} = \sum_{e \in E(v)} x_{e,t-1} \leq 0.8b_v$. On the other hand, if $\sum_{e \in E(v)} x_{e,t-1} \leq 0.4 b_v$, then $\sum_{e \in E(v)} x_{e,t} \leq \sum_{e \in E(v)} 2x_{e,t-1} \leq 0.8b_v$.
\end{proof}

\section{$(1 + \eps)$ Approximation of Unweighted $b$-matchings}
\label[section]{sec:1+eps-unweighted}
In this section, we prove the following result.
\begin{theorem}
\label{theorem:main-1+eps-unweighted}
    There exists an MPC algorithm that with high probability in $O(\log \log{\bar{d}} \cdot \exp(2^{O(1/\eps)}))$ rounds finds a $(1+\eps)$-approximate unweighted $b$-matching. This algorithm assumes that a machine has $\tO(n)$ memory and the algorithm uses a total memory of $\tO(m)$.  Here, $\bar{d}$ denotes the average degree of an input graph.
    Moreover, there exists a semi-streaming implementation of this algorithm that uses $\tO(\sum_v b_v + \poly(1/\eps))$ memory and performs $\exp(2^{O(1/\eps)})$ passes.
\end{theorem}

As the first step, in \cref{sec:existence-of-short-augmentations} we show that a $b$-matching which is not a $(1+\eps)$-approximate maximum $b$-matching can be substantially augmented via short augmenting walks.
We prove that claim by mapping a $b$-matching in $G$ to a $1$-matching in a new graph $G'$. In $G'$, each vertex $v \in V(G)$ is copied $b_v$ times. We also describe how to map $E(G)$ to $G'$ so that a $b$-matching in $G$ becomes a $1$-matching in $G'$. This enables us to carry over all $1$-matching properties to $b$-matchings, and in particular we obtain that an $O(1)$-approximate $b$-matching can be augmented to a $(1+\eps)$-approximate one via $O(1/\eps)$ long augmenting walks.

Despite this structural connection between $b$- and $1$-matchings, existing approaches for finding $(1+\eps)$-approximate $1$-matchings cannot be directly applied to $b$-matchings.
Namely, while our mapping enables us to map $b$-matched to $1$-matched edges, it does not provide a way of mapping \emph{unmatched} edges. This poses the main challenge in re-using existing augmenting methods for $1$-matchings in the context of $b$-matching.
Nevertheless, in \cref{sec:our-approach-1+eps-unweighted} we design an algorithm that overcomes this challenge for the approach designed in \cite{mcgregor2005finding}. 
An additional overview of our techniques is given in \cref{sec:overview-of-techniques}.

\subsection{Preliminaries}
As aforementioned, one of our key insights in this section is that $b$-matching on vertex set $V$ can be seen as a $1$-matching on a different vertex set $V'$. In this section we introduce two operations, \decompress and \compress, that map $V$ to $V'$ and map $V'$ back to $V$. An example is provided in \cref{fig:compress-decompress}.
\begin{figure}
	\centering
	\includegraphics[width=0.4\linewidth]{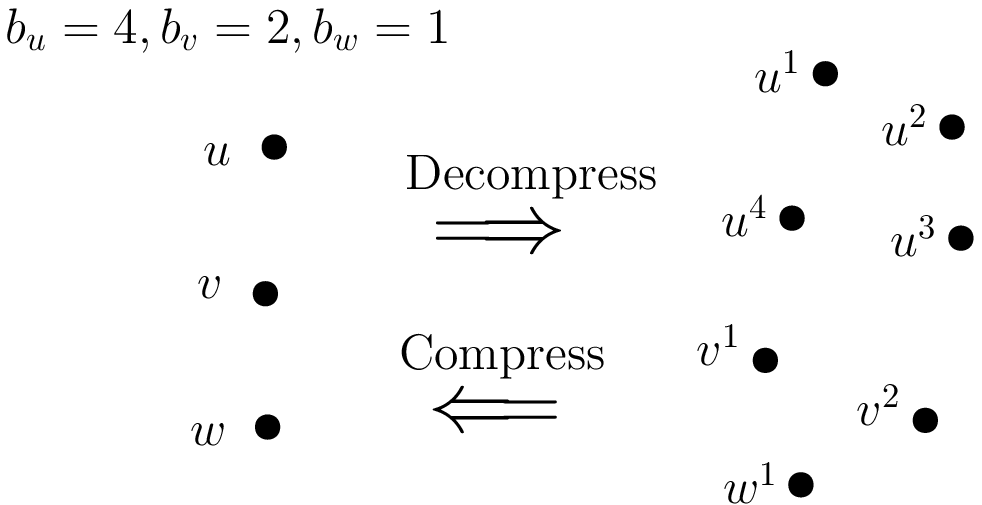}
	\caption{
	    An example of the operations described in \cref{definition:decompress,definition:compress}.
	}
	\label{fig:compress-decompress}
\end{figure}

\begin{definition}[\decompress operation]\label{definition:decompress}
    Given a set of vertices $V$ and a vector $b \in \bbN_{\ge 0}^V$ we define $\decompress(V, b)$ to be a set of vertices $V'$ such that for each $v$ the set $V'$ contains $b_v$ copies of $v$ denoted by $v^1, \ldots, v^{b_v}$. In particular, $|V'| = \sum_v b_v$.
\end{definition}

\begin{definition}[\compress operation]\label{definition:compress}
    Given a set of vertices $V'$, we define $\compress(V')$ to be a set of vertices $V$ such that $u \in V$ only if there exists an $i$ such that $u^i \in V'$. In particular, it holds that $\compress(\decompress(V, b)) = V$ for any $b \in \bbN_{\ge 1}^V$.
\end{definition}

\subsection{Proving the existence of short augmenting walks}
\label{sec:existence-of-short-augmentations}
Let $\Mopt$ be a maximum $b$-matching and $M$ be the matching we are aiming to augment. We will reduce this task to that of augmenting $1$-matchings in the following way.

We think of $M$ and $\Mopt$ as of two sets: $M \cap \Mopt$ and $M \triangle \Mopt$. Our goal is to keep $M \cap \Mopt$ ``intact'' and to augment $M$ to $\Mopt$ via alternating walks in $M \triangle \Mopt$. For the sake of brevity, define the set of edges $\Mdiff = M \triangle \Mopt$. Then, define $b'_v$ as 
\begin{align*}
    b'_v = \max( & \text{degree of $v$ in $\Mdiff \cap M$},\\
    & \text{degree of $v$ in $\Mdiff \cap \Mopt$}).
\end{align*}
As a direct consequence, let $M'$ be \emph{any} $b'$-matching in $\Mdiff$. Then, both $(M \setminus \Mdiff) \cup M'$ and $(\Mopt \setminus \Mdiff) \cup M'$ are $b$-matchings.

Given $b'$ we construct graph $H$, whose each edge corresponds to $M \cap \Mdiff$ or to $\Mopt \cap \Mdiff$, as follows:
\begin{enumerate}[label=(\Alph*)]
    \item $V(H) = \decompress(V(G), b')$.
    \item\label{item:add-M-edges} For each edge $\{u, v\} \in M \cap \Mdiff$, select a copy $u^i \in V(H)$ and a copy $v^j \in V(H)$ such that no $u^i$ nor $v^j$ has any edge from $M$ incident to it. Add edge $\{u^i, v^j\}$ to $E(H)$ and mark it as being in $M$.
    \item Repeat the same steps as in Step~\ref{item:add-M-edges} with $M$ being replaced by $\Mopt$.
\end{enumerate}
Observe that by the way we set $b'$ there always exist unmatched copies $u^i$ and $v^j$ to complete Step~\ref{item:add-M-edges}. By construction, $H$ restricted to the $M$-edges is a matching. Also, $H$ restricted to the $\Mopt$-edges is a matching. Hence, there exists a collection of alternating paths that augment the $M$-matching of $H$ to one having the same size as $\Mopt$-matching in $H$. Moreover, each of those alternating paths can be applied independently of the rest. These alternating paths correspond to alternating walks in $G$, where the correspondence is performed by mapping all copies $v^1, \ldots, v^{b'_v}$ of $H$ back to $v$ in $G$.
\\
\textbf{Remark:} For the sake of intuition we point out that $E(H)$ contains at most one edge between copies of $u$ and $v$ of $G$. This is the case as if $\{u, v\} \in M \cup \Mopt$ then $\{u, v\} \notin \Mdiff$, and $\{u, v\}$ appears at most once in $M$ and at most once in $\Mopt$.

Our discussion in this section enables us to apply the following result to $H$.
\begin{lemma}[\cite{Eggert2012}, Lemma 2]\label{lemma:certificate-for-maximality}
	Let $M$ be an inclusion-maximal matching. Let $\cY$ be an inclusion-maximal set of disjoint $M$-augmenting $k$-alternating-paths such that $|\cY| \le \tfrac{1}{k(k+2)} |M|$. Then, $M$ is a $(1 + 2/k)$-approximation of a maximum matching.
\end{lemma}

\subsection{Finding short augmenting walks for $1$-matchings}
\label{sec:finding-short-augmentations}
By our construction of $H$ in \cref{sec:existence-of-short-augmentations}, finding augmenting walks in $G$ can be seen as finding augmenting paths of the same lengths in $H$. However, $H$ is constructed based on $\Mopt$, which we do not have the access to. Therefore, we cannot directly use $H$ to find augmenting walks in $G$. Nevertheless, we still show how to find sufficiently many augmenting walks in $G$ by following the way $H$ is constructed in the context of the approach developed in \cite{mcgregor2005finding}. 

\paragraph{Recalling \cite{mcgregor2005finding}.}
Let $W$ be an undirected graph and $M$ a matching in $W$.
Based on \cref{lemma:certificate-for-maximality}, to find a $(1+\eps)$-approximate maximum matching in $W$, it suffices to consider augmentations of vertex-lengths in $O(1/\eps)$ in $W$. Let us fix an integer $2k + 1 \in O(1/\eps)$.

Next, we construct a layered graph that has $k + 2$ layers, denoted by $L_0, \ldots, L_{k + 1}$. A free vertex is uniformly at random assigned to $L_0$ or $L_{k + 1}$. Each matched edge $\{v, v'\}$ is assigned to one of the layers $L_1, \ldots, L_k$ uniformly at random as an arc $(v, v')$ or $(v', v)$; the arc direction is also chosen randomly. 
Each unmatched edge $\{x, y\}$ such that $(z, x) \in L_i$ (or $x \in L_i$ if $i = 0$) and $(y, z') \in L_{i + 1}$ (or $y \in L_{i + 1}$ if $i + 1 = k + 1$) is added an edge between the vertex $x$ in $L_i$ and the vertex $y$ in $L_{i + 1}$. Note that there might exists edges in $W$ that do not appear in the layered graph.
Intuitively, the layering \emph{guesses} the order of the edges in an augmentation. Observe that a fixed $(2k)$-length path appears in a layered graph with probability $1/(2k)^{k + 1} \in 1/\exp(O(1/\eps))$.
Then, the goal is to find many augmentations in this layered graph whose vertices pass through each of the layers exactly once.

To find many of those augmentations, Mcgregor~\cite{mcgregor2005finding} proposed an algorithm that iteratively ``grows'' (or extends) a set of alternating paths. 
Initially, this set of paths equals all the vertices in $L_0$. Then, as the next step, these paths are extended to vertices in $L_1$ such that the extensions remain disjoint, and to each path extended to $L_1$ is appended the corresponding matched edge. The algorithm continues in the same manner to extend the current alternating paths ending in $L_1$ to $L_2$ and so on, with an option to backtrack on all the current paths in case it is not possible to extend sufficiently many of them.

It is possible to show the following.
\begin{lemma}
\label{theorem:layering-claim}
    Consider a layered graph $\cL$ constructed with respect to a matching $M$.
    Let $\gamma |M|$ be the size of some maximal augmenting alternating paths in $\cL$.
    Let $\cP$ be a set of disjoint paths with their endpoints in layer $i$. Let $T$ be the maximum number of paths in $\cP$ that can be extended (in a vertex-disjoint manner) to layer $i+1$. If there is an algorithm $\cA$ that extends $\Theta(1) \cdot T$ paths of $\cP$ to layer $i+1$, then there is an algorithm that finds $(\gamma - \delta) |M|$ disjoint augmentations in $\cL$. Furthermore, this algorithm invokes $\cA$ for $O(\delta^{-2^{k}})$ many times, where $k$ is the number of layers in $\cL$.
\end{lemma}
Note that if $\cL$ in \cref{theorem:layering-claim} contains $\gamma |M|$ maximal augmenting paths, then the maximum number of augmenting paths in $\cL$ is at most $(2k + 2) \gamma |M|$.
Hence, it is instructive to think of $(2k + 2) \gamma$ in \cref{theorem:layering-claim} as the expected fraction of augmenting paths of length $2k + 1$ from $W$ that are preserved by layering.
\cref{theorem:layering-claim} is proved in \cite{mcgregor2005finding} for the case when $\cA$ finds a \emph{maximal} collection of extensions. In \cite{czumaj2019round} it was observed that the maximality can be replaced by findings a constant fraction of the maximum number of extensions.

\paragraph{Main challenges.}
There are two main challenges in  applying \cref{theorem:layering-claim} to $b$-matchings. First, given a graph $G$ for which we want to find a $b$-matching, it is natural to construct a graph $W$ similar to $H$ from \cref{sec:existence-of-short-augmentations}, i.e., to let $V(W) = \decompress(V(G), b)$, and hence reduce the problem of finding an approximate $b$-matching to the one of finding a $1$-matching. However, although our construction of $H$ gives a way to place matched edges between the vertices of $W$, it is not clear how to place unmatched edges.
Second, if we would use an algorithm for constructing approximate maximum matchings, as it is done in \cite{mcgregor2005finding}, to extend augmentation from one to the next layer, it would require $O(\sum_v b_v)$ memory per machine and $O(\log \log n)$ MPC rounds, or $\tO(n)$ memory per machine and $O(\sqrt{\log{n}})$ MPC rounds. On the other hand, our main result of this section, i.e., \cref{theorem:main-1+eps-unweighted}, achieves $O(\log \log n)$ round complexity with $\tO(n)$ memory per machine.

\subsection{Our approach}
\label{sec:our-approach-1+eps-unweighted}

We begin by invoking $\decompress(V(J), b)$ and placing matched edges between its vertices as described in \cref{sec:existence-of-short-augmentations}. MPC implementation of this step is provided in
\cref{lemma:assigning-M-between-decompress}.

To resolve the challenges outlined above, we proceed in three steps. Let $k + 2$ be the number of layers in a layered graph. As described, $L_i$, for $i = 1 \ldots k$, consists of matched edges that we view as arcs. Let $H_i$ be the heads and $T_i$ be the tails of the arcs in $L_i$. First, for each unmatched edge $e = \{u, v\} \in E(G)$ we choose an integer $i_e$ uniformly at random from the interval $0 \ldots k$ and an orientation $\arc{e} = (u, v)$ or $\arc{e} = (v, u)$; wlog, assume $\arc{e} = (u, v)$.
Then, we assign $e$ to be between a copy of $u$ which is in $H_{i_e}$ (if any) and a copy of $v$ which is in $T_{i_e + 1}$ (if any). However, and crucially, we \emph{do not} assign which copy of $u$ and which copy of $v$. We point out again that it is not clear how to choose copies while assuring that many augmentations in this process are retained.

Instead, as the second step, to find a large fraction of extension from $L_i$ to $L_{i+1}$ we \emph{contract} all copies of a vertex $u$ in layer $H_i$ to a single vertex, i.e., we invoke $\compress(H_i)$, and perform the same for the copies of $u$ in layer $T_{i+1}$, i.e., we invoke $\compress(T_{i + 1})$. Let $G_{j}$ be the graph obtained from layer $j$ by contracting copies of a vertex to a single vertex in layer $j$. Then, an unmatched edge $\arc{e} = (u, v)$, whose orientation is chosen randomly, where $i_e = i$, is added between vertex $u$ in $G_{i}$ and $v$ in $G_{i + 1}$.
As a remark, on orientation of an edge is chosen so to avoid having two copies of an edge $\{u, v\}$ between $G_{i}$ and $G_{i + 1}$ -- one copy between $u$ in $G_{i}$ and $v$ in $G_{i + 1}$, and the other copy between $v$ in $G_{i}$ and $u$ in $G_{i + 1}$.

As the final third step, we find a $\Theta(1)$ approximation of a maximum $b'$-matching between $G_{i}$ and $G_{i+1}$, where $b'_{w,j}$ is the number of copies of $w$ in layer $j$.

\paragraph{Analysis.}
Consider a maximum set of augmentations of length at most $2k+1$ in $H$. Then, an augmentation $P = (u_0 v_1 u_1 \ldots u_k v_{k + 1})$ is preserved in the corresponding layered graph if (1) the $i$-th matched edge of $P$ is assigned a label $i$ together with the right arc-orientation; (2) $u_0$ is assigned to $L_0$ and $v_{k + 1}$ is assigned to $L_{k+1}$; (3) to each edge $\{u_i, v_{i + 1}\}$ is assigned orientation $(u_i, v_{i + 1})$; and (4) for each $i$, the edge $\{u_i, v_{i + 1}\}$ is assigned to be between $L_i$ and $L_{i + 1}$. For a given augmentation $P$, all these events are true with probability at least $1/\exp(O(2k))$, where we let $k = O(1/\eps)$.
Therefore, the described process above in expectation preserves at least $1/\exp(O(1/\eps))$ fraction of a maximum collection of augmenting walks of size at most $2k+1$ in $G$.
This together with \cref{theorem:layering-claim} implies the following result.
\begin{lemma}
\label{lemma:unweighted-1+eps-number-of-invocations}
    Let $\cA$ be an algorithm that with high probability computes a $\Theta(1)$-approximate maximum $b$-matching. Then, there exists an algorithm $\cB$ that by invoking $\cA$ $\exp(2^{O(1/\eps)})$ many times in expectation outputs a $(1+\eps)$-approximate $b$-matching.
\end{lemma}
We now provide more details on how \cref{theorem:layering-claim} is invoked to prove \cref{lemma:unweighted-1+eps-number-of-invocations}.
Let $\gamma'$ be (an upper-bound on) the probability that a ``short'' augmentation exists in a layered graph; as discussed, $\gamma' \in 1/\exp(O(1/\eps))$. Since we are interested in a $(1 + \eps)$-approximate maximum $b$-matching, our goal is to improve the current matching as long as there are at least $\eps |\Mopt|$ disjoint augmenting paths, i.e., as long as a maximal number of ``short'' augmenting paths is at least $\eps |\Mopt| / (2k + 1) = c \eps^2 |\Mopt|$, for an appropriately chosen constant $c$. Hence, we invoke \cref{theorem:layering-claim} by letting $\gamma = \gamma' \cdot c \cdot \eps^2$ and $\delta = \gamma/2$.

Observe that the claim for $\cB$ in \cref{lemma:unweighted-1+eps-number-of-invocations} is given in expectation. Nevertheless, using standard techniques, and executing $O(\log n)$ copies of $\cB$ in parallel, it is possible to obtain the same guarantee with high probability in the following way. \cref{lemma:unweighted-1+eps-number-of-invocations} claims that the resulting matching is in expectation an $\eps$ fraction smaller than the maximum one. Markov's inequality implies that with probability at least $1/2$ the resulting matching is $(1 + 2 \eps)$-approximate. Instead of executing $\cB$ only once, consider an execution of $c \log n$ independent copies of $\cB$ for $\eps' = \eps/2$, where $c$ is a constant; these copies are executed in parallel. Then, the largest matching among those $c \log{n}$ copies is a $(1+\eps)$-approximate one with probability at least $1 - n^{-c}$.
Finally, this discussion together with the MPC implementation in \cref{sec:MPC-implementation-1+eps-unweighted} and the streaming implementation in \cref{sec:streaming-implementation-1+eps-unweighted} concludes the proof of \cref{theorem:main-1+eps-unweighted}.


\subsection{MPC implementation}
\label{sec:MPC-implementation-1+eps-unweighted}
\begin{lemma}
\label{lemma:assigning-M-between-decompress}
    The process of assigning matched edge to $\decompress(V, b)$ such that to a copy of a vertex is assigned at most one matched edge, as done in \cref{sec:existence-of-short-augmentations}, can be done in $O(1)$ MPC rounds with $O(n^{\delta})$ memory per machine, for any constant $\delta > 0$.
\end{lemma}
\begin{proof}
    Consider a matching $M$. For each edge $e = \{u, v\}$ we create two pairs $(u, e)$ and $(v, e)$. To each such pair $(v, e)$ our goal is to assign an integer $i_{v, e}$ between $1$ and $b_v$ meaning that edge $e$ is assigned to the $i_{v, e}$-th copy of $v$. Two pairs $(v, e)$ and $(v, e')$ should be assigned distinct number.
    
    To obtain $i_{v, e}$, first we sort all these pairs. Then, we compute prefix sum where each pair is treated as having value $1$. Sorting and prefix sum can be performed in $O(1)$ MPC rounds, as explained in \cite{Goodrich2011SortingSA}.
    
    Focus now on the sequence of pairs $(v, e_1), (v, e_2), \ldots, (v, e_t)$. Their prefix sums are $S_v, S_v+1, \ldots, S_v+t$. However, only the very first pair knows $S_v$, while pair $(v, e_j)$ sees only the sum $S_v + (j - 1)$. Our aim is that all the pairs learn $S_v$. To that end, we label the pair $(v, e_1)$ as a special one and create a search tree over all special pairs. Then, each non-special pair $(v, e_j)$ queries the search tree for $(v, e_1)$ and its prefix sum value. When $(v, e_j)$ receives $S_v$, it assign $i_{v, e_j} = (S_v + j) - S_v + 1$. This search tree can be simulated in $O(1)$ MPC rounds, as described in \cite[Theorem 4.1]{Goodrich2011SortingSA}.
\end{proof}

\subsection{Streaming implementation}
\label{sec:streaming-implementation-1+eps-unweighted}
We now discuss a semi-streaming implementation of $(1+\eps)$-approximate unweighted $b$-matchings, i.e., the algorithm $\cB$ from \cref{lemma:unweighted-1+eps-number-of-invocations}. For the sake of brevity, we let $B = \sum_v b_v$.

First, in $O(B)$ space we store all matched edges and all the information associated with them. To find a maximal $b'$-matching between two layers (which is used to extend alternating paths from one to the next layer as discussed in \cref{sec:our-approach-1+eps-unweighted}) the algorithm simply scans over unmatched edges and adds them to the between-layers-matching greedily. This greedy algorithm is a $2$-approximate one.

The main challenge in implementing this approach in $O(B)$ space is the fact that each unmatched edge carries information as well, which is its orientation and the layer it has been assigned to; this assignment has been discussed in \cref{sec:our-approach-1+eps-unweighted}. Unfortunately, this is $O(m)$ amount information, which can be much larger than $B$. Nevertheless, we observe that our analysis in \cref{sec:our-approach-1+eps-unweighted} requires that a certain number of augmenting paths appears in a layered graph only in expectation. Hence, as long as the orientation and a layer assignment for each edge of an augmentation is chosen independently, the same analysis we provided in \cref{sec:our-approach-1+eps-unweighted} carries over.
Since an augmenting path in a layered graph has length $\poly(1/\eps)$, instead of using independent random bits for every single unmatched edge to assign orientation and layers, in the streaming setting we use $k$-wise independent hash function for $k \in \poly(1/\eps)$. Such function can be constructed in space $O(k \log m) = O(\poly(1/\eps) \cdot \log m)$ space, as provided by the following result.
\begin{theorem}[\cite{alon1986fast}]
\label{theorem:t-wise-independent}
    For $1 \le t \le L$, there is a construction of $t$-wise independent random bits $x_1, \ldots, x_L$ with seed length $O(t \log L)$. Furthermore, each $x_i$ can be computed in space $O(t \log L)$.
\end{theorem}
We instantiate \cref{theorem:t-wise-independent} with $t = O((\log{n} + \log 1/\eps) \cdot k)$ -- for each edge $e$ of a $k$-length augmenting path it suffices to use $O(\log n + \log 1/\eps)$ random bits to assign a random layer number to $e$. We set $L = t \cdot m$.
Given that our main algorithm in \cref{sec:our-approach-1+eps-unweighted} executes $O(\log n)$ independent copies of $\cB$ in parallel, the total memory requirement for the semi-streaming setting is $\tO(B + \poly(1/\eps))$.
\section{$(1+\eps)$-approximate Weighted $b$-matchings: Outline}
\label[section]{sec:weighted-b-matching}

In this section we focus on weighted $b$-matchings.
\begin{theorem}[Restated \cref{them:main-weighted-eps-constant}]\label{theorem:main-weighted}
    Let $G$ be a weighted graph on $n$ vertices, and let $\eps > 0$ be a parameter. When the memory per machine is $\tO(n + \poly(1/\eps))$, there exists an MPC algorithm that in $O(\log \log{\bar{d}} \cdot \exp(2^{O(1/\eps)}))$ rounds outputs a $(1+\eps)$-approximate $b$-matching. The approximation guarantee holds with probability $1-1/n^c$, where $c$ is an arbitrary constant. This algorithm uses the total memory of $\tO(m \cdot \exp(1/\eps))$.
    Here, $\bar{d}$ denotes the average degree of the graph.
    Moreover, there exists a semi-streaming implementation of this algorithm these uses $\tO(\exp(1/\eps) \cdot \sum_v b_v)$ memory and $\exp(2^{O(1/\eps)})$ passes.
\end{theorem}

\paragraph{Organization and overview of the analysis.}
Similar to our approach in \cref{sec:1+eps-unweighted}, we aim to build on the connection described in \cref{sec:existence-of-short-augmentations} and reuse existing results for finding $(1+\eps)$-approximate weighted $1$-matchings. This approach has several challenges that we outline next.

We build on the result of \cite{gamlath2019weighted}, which shows how to compute $(1+\eps)$-approximate weighted $1$-matchings. That result reduces a computation of $(1+\eps)$-approximate weighted to a computation of $O(\log n \cdot \exp(1/\eps^2))$ many $(1+\eps)$-approximate unweighted $1$-matchings. We show how to obtain a similar reduction in the context of $b$-matchings, which will also enable us to leverage the result from \cref{sec:1+eps-unweighted}.

To describe the first challenge we encounter, we recall that the prior work creates a random bipartite graph from the input graph. A bipartite graph is created by randomly placing a vertex to one of the two partitions, and then only edges connecting the two partitions are kept while the remaining ones are ignored. Applying the same idea in the context of $b$-matchings directly means that we create $b_v$ copies of each vertex $v$ and perform the described process of randomly placing each vertex-copy to one of the partitions. However, it might happen that two copies of $v$ appear in different partitions, and hence it is unclear how to place their incident edges while preserving bipartiteness. To understand why bipratiteness is helpful, we also recall that the prior work finds augmenting paths for $1$-matchings in two stages: first, it finds special kind of alternating \emph{walks}, i.e., vertices and even edges might repeat in these alternations; and, second, it decomposes each alternating walk into a path and simple cycles.
Bipartiteness ensures that each cycle has even length, and hence it is possible to augment over it. Nevertheless, despite the fact that copies of the same vertex might belong to different partitions, we show how to achieve these desired properties in the context of $b$-matchings by orienting edges.
These ideas, together with an example, are in detail discussed in \cref{section:layering-for-b-matchings}.

Second, the approach developed in \cite{gamlath2019weighted} requires local memory of $O(|\Mopt|)$, where $|\Mopt|$ is the size of a maximum matching. For $b$-matchings, the same algorithm requires $O(\sum_v b_v)$ memory, which in general case is asymptotically much larger than the memory per machine of $\tO(n + \poly 1/\eps)$ that we are considering. The reason for this requirement is that a collection of computed augmenting walks might overlap, so it is needed to resolve these conflicts while ensuring that the remaining walks still lead to large-in-weight augmentation. The prior work performs that by gathering all the walks on a single machine and resolving conflicts on that machine locally.
To alleviate that, we develop a parallel approach for conflict resolution that requires very little memory per machine, i.e., only $n^\delta$ memory for any constant $\delta$. This approach is in detail presented 
\cref{sec:conflict-resolution}.

In \cref{sec:overview-of-techniques} we provide an additional overview of our techniques used to prove this theorem.

\subsection{Preliminaries}

\begin{definition}[Applying a walk]
    Given a matching $M$ in a weighted graph $G$, let $P$ we an alternating walk in $G$. By \emph{applying} $P$ we refer to the operation of replacing $M$ by $(M \setminus (M \cap P)) \cup (M \triangle P)$, i.e., removing from $M$ the matched edges in $P$ and adding to $M$ the unmatched ones.
\end{definition}

\begin{definition}[$\gain$]
    Given a matching $M$ in a weighted graph $G$, let $P$ we an alternating walk in $G$. We use $\gain(P)$ to refer to the weight-difference between the unmatched and the matched edges of $P$, i.e.,
    $\gain(P) = w(P \triangle M) - w(P \cap M)$.
    Intuitively, $\gain(P)$ represents the weight-increase that would be obtained if $P$ is applied.
\end{definition}

\subsection{Graph layering for weighted $1$-matchings}
\label{section:recall-layering-weighted}
We now recall the main parts of graph layering in \cite{gamlath2019weighted}, which is designed to reduce the problem of approximating weighted $1$-matching to that of approximating unweighted $1$-matchings. In \cref{section:layering-for-b-matchings} we describe how to make a step further and obtain a similar reduction for $b$-matchings, while assuring it is an efficient reduction.

A layered graph is defined as a function of a weight $W$, a weighted \emph{bipartite} graph $J'$, and two threshold sequences $\tau^A$ and $\tau^B$ (that will be explained soon). We use $\cL(J', W, \tau^A, \tau^B)$ to denote such layered graph. Intuitively, $\cL(J', W, \tau^A, \tau^B)$ captures short weighted augmentations in $J'$ whose total weight is $\Theta(W /\poly(\eps) )$ and the gain of each of those augmentations, if applied, is $\Theta(W \poly(\eps))$.

\begin{enumerate}[label=\arabic*.]
    \item
    Given a weighted graph $J$ (not necessarily bipartite), a matching $M$ in $J$, as the first step of layering procedure is to create two partitions of vertices -- let those partitions be called $H$ and $T$. Each vertex is assigned to $H$ or $T$ uniformly at random. Then, each edge between two vertices in the same partition is ignored. This yields a bipartite graph $J'$. We say that $J'$ is obtained by \emph{bipartiting} $J$.

    \item
    A layered graph $\cL(J', W, \tau^A, \tau^B)$ (which is different than the one described in \cref{sec:1+eps-unweighted}) is created from $J'$ as follows. $\cL(J', W, \tau^A, \tau^B)$ consists of $k + 1$ layers. \emph{Each} layer contains a copy of $H$ and a copy of $T$, i.e., the vertex set of $J'$ is repeated in $\cL(J', W, \tau^A, \tau^B)$ $k+1$ times, where vertex $v \in V(J')$ in layer $i$ of $\cL(J', W, \tau^A, \tau^B)$ is denoted as $v_i$. Let $H_i$ denote the copy of $H$ and $T_i$ denote the copy of $T$ in the $i$-th layer.

    \item
    To each layer $i$ is assigned a threshold $\tau_i^A$ and between two layers $i$ and $i + 1$ is assigned threshold $\tau_i^B$.

    \item\label{item:weights-close-to-tauA-and-tauB}
    For each $i$, a matched edge $e = \{u, v\}$ is added between $u_i$ and $v_i$, assuming that $u \in H$ and $v \in T$, only if $w(e)$ is ``close'' to  $\tau_i^A W$.\footnote{Here, close means that $w(e) \in [\poly(\eps) \tau^A_i W, (1 + \poly(\eps)) \tau^A_i W]$. We point a reader to \cite[Section 4.3 of the arXiv version]{gamlath2019weighted} for details.}
Similarly, an unmatched edge $e = \{u, v\}$, such that $u \in H$ and $v \in T$, is added between $u_i$ and $v_{i + 1}$ only if $w(e)$ is ``close'' to $\tau^B_i W$. Therefore, matched edges are placed within layers, while unmatched edges are placed between neighboring layers.

    
    \item\label{item:cleaning-step}
    Remove each vertex in the layered graph $\cL$ which is not incident to a matched edge in $\cL$, where matched vertices in the first and the last layer are treated specially. That is, if a free vertex $v$ is in $H_1$ and $\tau^A_1 = 0$, keep $v$; otherwise remove $v$. Similarly, if a free vertex $v$ is in $T_{k + 1}$ and $\tau_{k+1}^A = 0$, keep $v$; otherwise remove $v$.

    \item\label{item:tauA-tauB-multiples}
    The thresholds $\tau^A$ and $\tau^B$ are chosen in such a way that each their coordinate is a multiple of $\eps^{12}$, $\sum_i \tau^B_i - \sum_i \tau^A_i \ge \eps^{12}$ and $\sum_i \tau^B_i \le 1 + \eps^4$. That is, each 
    alternating path $P$ passing through \emph{all} the layers (which also means $P$ has its first and last edge in $M$) has sufficiently large gain.
\end{enumerate}

The following claim follows from Algorithms 3 and 4 and Theorem 4.8 of the arXiv version of \cite{gamlath2019weighted}.
\begin{theorem}[\cite{gamlath2019weighted}, rephrased]\label{theorem:gamlath}
    Assume that the following algorithms exist:
    \begin{enumerate}[label=(\Alph*)]
        \item
        Let $\AlgLayered$ be an algorithm for finding a $(1 + \delta)$-approximate \emph{unweighted} maximum matching in a bipartite graph.
        
        \item\label{item:AlgAlternating}
        Given a blue $\Mblue$ and a red matching $\Mred$ where $\Mblue \cap \Mred = \emptyset$ and $\Mblue \cup \Mred$ does not contain cycles. Moreover, exactly one vertex of some red edges is marked as \emph{special} and it is guaranteed that a special vertex is not incident to a blue edge.
        Let \AlgAlternating be an algorithm that finds red-blue alternating paths in $\Mblue \cup \Mred$ where each such path begins with a special vertex and each path is as long as possible.
        
        \item\label{item:Alg-Resolve-oracle}
        Let $\cW$ be a set of \emph{distinct} weights such that each $W \in \cW$ is of the form $(1+\eps^4)^i$ for some $i \in \bbN_{\ge 0}$.
        For $W_i \in \cW$, let $\cP_{W_i}$ be a collection of vertex disjoint alternating paths in $\cL(J', W_i, \tau^A_i, \tau^B_i)$, where $J'$ is obtained by bipartiting $J$ and each path passes through all the layers. Let \AlgResolve be an algorithm that by using augmentations $\bigcup_i \cP_{W_i}$ finds a collection $\cA$ of edge disjoint augmentations in $J$ whose gain is at least $\rresolve \cdot \sum_{i} \sum_{P \in \cP_i} \gain(P)$.
    \end{enumerate}

    Then, there exists an algorithm that finds a $(1+\eps)$-approximate maximum weighted matching by invoking $\AlgLayered$, $\AlgAlternating$ and $\AlgResolve$ for $O(\exp(1/\eps^2) \cdot 1/\delta \cdot 1/\rresolve)$ times, where $\delta$ is of the order of the probability that a ``short'' augmenting walk in $J$ appears in a layered graph.
\end{theorem}

Our main task is now to implement $\AlgLayered$ and $\AlgResolve$. Before doing that, we first elaborate how these methods are used in the context of finding weighted $1$-matchings. In the same way, we will use them to find approximate weighted $b$-matchings.
\\
As a reminder, each alternating path passing through all the layers $\cL(J', W, \tau^A, \tau^B)$ is by construction such that it contains an augmentation which yields $\poly(\eps) W$ increase in the matching weight. Hence, a goal is first to find many such disjoint alternating paths in $\cL(J', W, \tau^A, \tau^B)$. Those paths might intersect when translated to $J'$, though. To find many such paths, we consider $\cL'$ obtained from $\cL(J', W, \tau^A, \tau^B)$ by removing the edges in the first and the last layer. Recall that those edges are matched. Then, in $\cL'$ we find a $(1+\delta)$-approximate (unweighted) maximum matching $M'$, for some $\delta \in \exp(1/\eps)$.
The idea now is that the union of $M'$ and $M$ contains many path whose edges alternate between $M$ and $M'$. To see why that holds, observe that an alternating path of $\cL$ passing through all the layers correspond to an alternating path in $\cL'$ that begins by and ends by an unmatched edge. Other than these alternating paths, due to Step~5, there are no additional odd-length alternating paths in $\cL'$ that begin with an unmatched edge.
Therefore, if we assume that there are sufficiently many-in-weight augmentations in $\cL'$, a close to maximum matching $M'$ in $\cL'$ is expected to select many unmatched edges. It can be shown that $M'$ combined with the $M$-edges appearing in $\cL(J', W, \tau^A, \tau^B)$ contain sufficiently many alternating paths that pass through all the layers of graph $\cL(J', W, \tau^A, \tau^B)$. In the context of \cref{theorem:gamlath}, we think of $\Mblue$ as the $M$-edges in $\cL'$ and of $\Mred$ as $M' \setminus M$. The special vertices is $H_1$.

The augmentations found in the previous step might intersect. Moreover, they might intersect when mapped to $J$ even if they are coming from the same layered graph. This is addressed by extracting a subset of non-intersecting augmentations.\footnote{The notion of ``non-intersecting'' is direct for $1$-matchings, but requires a special care in the context of $b$-matchings. Namely, multiple augmentations may pass through the same vertex $v$ as long as there are at most $b_v$ of them and they are edge-disjoint.} The prior work extracts a subset of non-intersecting augmentations greedily and sequentially, starting from the heaviest augmentations. We design a different method, which is parallel and scalable.

\subsection{Layering for $b$-matchings}
\label{section:layering-for-b-matchings}
Let $G$ be a weighted graph. Let $M$ be the current $b$-matching in $G$ which is not a $(1+\eps)$-approximate one. We now describe how we adapt the layering described in \cref{section:recall-layering-weighted} to $b$-matchings. We provide an example of $b$-matching layering in \cref{fig:layering}.

\begin{figure}
	\centering
	\includegraphics[width=0.2\linewidth]{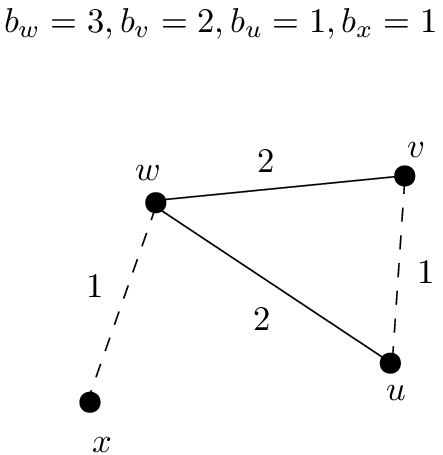}
	\caption{
	    An example of a weighted graph used in \cref{fig:layering} and a $b$-matching in that graph. The dashed edges are those in the matching.
	}
	\label{fig:example}
\end{figure}

\begin{figure}
	\centering
	\includegraphics[width=0.6\linewidth]{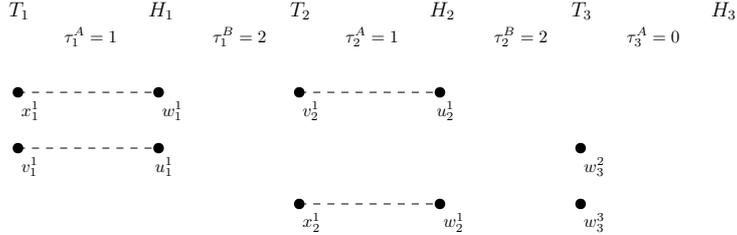}
	\caption{
	    An example a graph layering for $b$-matching for the graph $G = (V, E)$ shown in \cref{fig:example}. Dashed edges in the figure are those in the current matching.
	    Recall that a layering for $b$-matching is created on $\decompress(V, b)$. The matched edge $\{x, w\}$ is placed between $x^1$ and $w^1$, while $\{u, v\}$ is placed between $u^1$ and $v^1$.
	    Assume that in this layering $\{w^1, u^1, v^2\} \in H$ and $\{w^2, w^3, x^1, v^1\} \in T$. Note that $v^2$ does not appear in the layering. It is the case as $v^2 \in H$, but $v^2$ is not incident to a matched edge and also $\tau^A_1 \neq 0$; see Step~5 in \cref{section:recall-layering-weighted}. We also point out that unmatched edges $\{w, v\}$ and $\{w, u\}$ do not appear in the layering, as stated under (II) in \cref{section:layering-for-b-matchings}. However, those unmatched edges will be used inside corresponding layers by the algorithm to find augmentations.
	}
	\label{fig:layering}
\end{figure}

{


    \paragraph{(I)} Inspired by the connection between $1$- and $b$-matchings that we described in \cref{sec:existence-of-short-augmentations}, we consider $\decompress(V(G), b)$ and distribute each edge $\{u, v\} \in M$ between a copy of $u$ and a copy of $v$ in $\decompress(V(G), b)$ such that no copy of a vertex is incident to more than two edges of $M$. Let $\tM$ be the resulting matching, the one distributed between the vertices $\decompress(V(G), b)$ as described.

    \paragraph{(II)}
    We create a weighted layered graph $\cL$ with the vertex set $\decompress(V(G), b)$ and assign to it $\tM$ as described in \cref{section:recall-layering-weighted}, but we do not yet add unmatched edge to $\cL$.

    \paragraph{(III)}
    For each unmatched edge $e = \{u, v\} \in E(G)$ we choose a random orientation to be $\vec{e} = (u, v)$ or $\vec{e} = (v, u)$; this is akin to our approach described in \cref{sec:our-approach-1+eps-unweighted}. Without loss of generality, assume that $\vec{e} = (u, v)$. The meaning of this is that $e$ is allowed to connect in $\cL$ only a copy of $u$ in $H_i$ and a copy of $v$ in $T_{i + 1}$; but never a copy of $v$ in $H_j$ and a copy of $u$ in $T_{j + 1}$.
    
    We now elaborate on why this step is crucial. Let $A'$ be an augmentation in $\cL(J', W, \tau^A, \tau^B)$ and let $A$ be an alternating walk ``mapped back'' to $J$, i.e., a vertex $v_i \in H_i \cup T_i$ is replaced by $v$. Even for $1$-matchings it might happen that $A$ intersects itself. This comes from the fact that vertices, and hence the edges as well, repeat across different layers. Nevertheless, it can be shown that in the case of $1$-matching the walk $A$ can be decomposed into a union of even cycles and a single path. In the case of $1$-matching, these cycles and the path by construction do not have repeated edges.
    
    However, this property is not  ensured in the case of $b$-matchings. For instance, assume that $\cL$ contains an alternating path $P_{\cL}$ where $P_{\cL} = v^1_{H,1} u^1_{T,2} w^1_{H,2} x^1_{T,3} u^2_{H,3} v^2_{T,4}$ and both $v^1$ and $v^2$ are free. The issue with $P_{\cL}$ is that the same edge in $J$, the edge $\{u, v\}$ in this case, is accounted twice in $\gain(P_{\cL})$, which we want to avoid and ensure that a walk never goes over the same edge twice.
    As we show, this is achieved by performing random orientation as we described. 
This step is crucial for proving
\cref{lemma:proof-algorithm-extract-alternations}.

}



\subsection{MPC implementation of $\AlgAlternating$}
\label{alg:AlgAlternating}
We now describe the MPC implementation of \AlgAlternating and prove the following claim.
\begin{lemma}
    Given two disjoint matchings $\Mred$ and $\Mblue$ as described in \cref{theorem:gamlath}, point~(B), 
    let $k$ be the maximum alternating-path-length in $\Mblue \cup \Mred$ starting at a special vertex. Then, there exists an MPC algorithm that implements $\AlgAlternating$ in $O(k)$ rounds and uses $O(n^\alpha + k)$ memory per machine, for any constant $\alpha > 0$.
\end{lemma}
Our MPC implementation simply performs a BFS in parallel over all the special vertices. More formally, for each special vertex is created a path of vertex-length $1$. Then, the algorithm step by step extends all current paths by an edge from $\Mred$ (in odd-numbered steps) or extends all current paths by an edge from $\Mblue$ (in even-numbered steps). The first step is numbered $1$. Note that each vertex is incident to at most one edge from $\Mred$ and at most one edge from $\Mblue$. So, if a path can be extended, then that extension is unique.

\subsection{From layered alternating paths to augmenting walks}
\label{sec:;ayered-alternating-to-augmenting-walks}

\begin{algorithm}[h]
\TitleOfAlgo{\AlgExtracintAlternations}
    \SetAlgoLined
    \KwData{An alternating path $P_{\cL}$ in $\cL(J', W, \tau^A, \tau^B)$ passing through all the layers, where $J'$ is obtained by bipartiting $J$. If an endpoint of $P_{\cL}$ is in $H_1$, then the endpoint is free. If an endpoint of $P_{\cL}$ is in $T_{|\tau^A|}$, then the endpoint is free.}
    
    For each vertex $v \in V(J)$, let $\idx(v, H, P)$ be the value (in case of ties, choose any) such that $v^{\idx(v, H, P)}_{H, i}$ appears in $P_{\cL}$. Similarly, let $\idx(v, T, P)$ be the value (in case of ties, choose any) such that $v^{\idx(v, T, P)}_{T, i}$ appears in $P_{\cL}$.
    
    Replace each vertex $v^t_{H, i} \in P_{\cL}$ by $v^{\idx(v, H, P)}_{H}$ and each vertex $v^t_{T, i} \in P_{\cL}$ by $v^{\idx(v, T, P)}_{T}$. Let the resulting walk be $P'$.
    \label{line:obtain-P'}

    Since $P'$ is bipartite, it can be decomposed into a collection of even-length non-intersecting alternating cycles and a single non-intersecting alternating path. Let such $\cC'$ be one such collection of cycles and a path.
    
    Let $\cC \gets \emptyset$.
    
    \For{each $c' \in \cC'$} {
        Let $c$ be obtained from $c'$ by replacing each $v^t_{H}$ and each $v^t_{T}$ by $v$. \tcc*{Mapping $c'$ to $J$.}\label{line:map-c'-to-J}
        
        Add $c$ to $\cC$.
    }
    
    \Return{$\cC$}
\caption{An algorithm that transforms an alternating path $P_\cL$ in a weighted layered graph $\cL(J', W, \tau^A, \tau^B)$ to a collection of even-length walk-cycles and a single walk $c$ in $J$. Each endpoint $v$ of $c$ either corresponds to a copy of $v$ which is free, or the edge in $c$ incident to $v$ is matched. Moreover, no edge repeats twice in any of these walks.\label{alg:extracting-alternations}}
\end{algorithm}

\begin{lemma}
\label{lemma:proof-algorithm-extract-alternations}
    Let $\cC$ be the output of \cref{alg:extracting-alternations} for an input $P_{\cL}$. 
    Then,
    \begin{enumerate}[label=(\arabic*)]
        \item All but one walk of $\cC$ is even-length cycle.
        \item No edge in a walk of $\cC$ repeats.
        \item\label{item:c-non-cycle} Let $c$ be the walk (if any) in $\cC$ which is not an even cycle.
        Each endpoint $v$ of $c$ either corresponds to a copy of $v$ which is free, or the edge in $c$ incident to $v$ is matched. If both endpoints of $c$ are free, then they do not correspond to the same copy in $\decompress(V(G), b)$.
    \end{enumerate}
\end{lemma}
\begin{proof}
    Let $\cL(J', W, \tau^A, \tau^B)$ and $P_{\cL}$ be the input to \cref{alg:extracting-alternations}.

    The property that all but one walk in $\cC$ is even-length cycle follows from the fact that $P'$ constructed on \cref{line:obtain-P'} is bipartite. Hence, $\cC'$ contains even-length alternating cycles and a single alternating walk. When $\cC'$ is mapped to $J$ on \cref{line:map-c'-to-J}, we obtain a collection of even-length cycle-walks and a single walk $c$.
    Also, the endpoints of $c$ are the same to those of $P_{\cL}$. In particular, of both endpoints of $P_{\cL}$ are free, then one of them is in $H_1$ and the other is in $T_{|\tau^A|}$; hence, these endpoints differ. Therefore, the input constraints on $P_{\cL}$ imply Property~\ref{item:c-non-cycle}.
    
    Consider $c \in \cC$. We now want to show that no edge appears twice in $c$. Towards a contradiction, assume that $\{u, v\}$ appears twice in $c$. Choose a starting vertex of $c$ (if $c$ is not a cycle-walk, choose its first vertex in $c'$ as the starting one), and orient all edges from that starting vertex. We consider two cases.
    
    \paragraph{Both $(u, v)$ and $(v, u)$ appear.}
    Assume first that $\{u, v\}$ is matched, i.e., it belongs to $\tM$ constructed in Step~(I) of \cref{section:layering-for-b-matchings}. Hence, if both $(u, v)$ and $(v, u)$ are in $c$ it would imply that $u \in H_i$ and $u \in T_j$ for some $i$ and $j$, which is a contradiction with our construction of $\tM$.
    
    Hence, assume that $\{u, v\}$ is an unmatched edge. This is now in a contradiction with Step~(III) described in \cref{section:layering-for-b-matchings}, i.e., it cannot happen that $\cL$ contains both edge $(u^{t_1}_{H, i}, v^{t_2}_{T, i + 1})$ and edge $(v^{t_3}_{H, j}, u^{t_4}_{T, j + 1})$, where $t_1$, $t_2$, $t_3$ and $t_4$ do not have to be necessarily distinct.
    
    \paragraph{The orientation $(u, v)$ appears twice.} Assume that $(u, v)$ is matched; the case $(u, v)$ being unmatched follows analogously. This implies that $u$ appears as $u_{H, i}^{t_1}$ and $u_{H, j}^{t_2}$ in $P_{\cL}$ for some $i, j, t_1$ and $t_2$.
    However, both $u_{H, i}^{t_1}$ and $u_{H, j}^{t_2}$ are represented by the same vertex $u^{\idx(v, H, P)}_H$ in $P'$ (see \cref{line:obtain-P'} of \cref{alg:extracting-alternations}). Moreover, by the construction of layered graphs, two appearances of $u^{\idx(u, H, P)}_H$ are at even distance. Therefore, the even cycle containing the first appearance of $u^{\idx(u, H, P)}_H$ ``closes'' before the second appearance of $u^{\idx(u, H, P)}_H$. So, it does not contain twice $(u^{\idx(u, H, P)}_H, v^{\idx(v, T, P)}_T)$, and hence $c$ does not contain two oriented copies of $(u, v)$.
        
\end{proof}

\subsection{Implementation of $\AlgResolve$}
\label{sec:conflict-resolution}
Our implementation of $\AlgResolve$ (which is defined in \cref{theorem:gamlath}~(C) has two main components: resolving conflicts within a layered graph (that we describe in \cref{section:resolve-within-layered}); and, resolving conflicts between different layered graphs.

\subsubsection{Resolving conflicts within layered graph}
\label{section:resolve-within-layered}
    We now describe an algorithm that performs within-layered-graph conflict resolution.

 \begin{algorithm}[h]
\TitleOfAlgo{\AlgResolveWithinLayered}
    \SetAlgoLined
    \KwData{Let $\cP_W$ be a collection of vertex-disjoint alternating paths in $\cL(J', W, \tau^A, \tau^B)$, where $J'$ is obtained by bipartiting $J$ and each path passes through all the layers. 
    If an endpoint of $P \in \cP_W$ is in $H_1$, then the endpoint is free. If an endpoint of $P$ is in $T_{|\tau^A|}$, then the endpoint is free.}
    
    Let $\cA' \gets \emptyset$.
    
    \For{each $P \in \cP_W$}{
        With probability $1 - \eps^9/2$ skip processing $P$.\label{line:reject-in-resolve-within-layered}
    
        Let $\cC \gets \AlgExtracintAlternations(P)$.
        
        Let $C$ be an alternating walk from $\cC$ with the largest gain.\label{line:select-largest-gain-resolve-within-layered}
        
        Add to $C$ to $\cA'$.
    }
    
    Let $\cA \gets \emptyset$
    
    For a walk $C$ in $J$, we use $\decompress_{\cap C}(V(J), b)$ to denote vertices of $\decompress(V(J), b)$ that correspond to those in $C$.
    \label{line:define-decompress-cap-C}
    
    \For{each $C \in \cA'$}{
        \If{$\decompress_{\cap C}(V(J), b) \cap \decompress_{\cap C'}(V(J), b) = \emptyset$ \textsc{and} $E(C) \cap E(C') = \emptyset$, for every $C' \in \cA$ such that $C \neq C'$\label{line:test-intersection-resolve-within-layered}} {
            Add $C$ to $\cA$.\label{line:add-C-to-cA-rsolve-within-layered}
        }
    }
    
    \Return{$\cA$}
    
\caption{An algorithm that given a collection of vertex-disjoint augmentations in a layered graph $\cL(J', W, \tau^A, \tau^B)$, where each augmentation passes through all the layers, extracts a collection of edge-disjoint augmenting walks in $J$. Moreover, when these walks are mapped to $\decompress(V(J), b)$ they are vertex-disjoint as well.
\label{alg:resolving-conflicts-in-layered}}
\end{algorithm}

\begin{lemma}\label{lemma:resolve-within-layered}
    Let $\cP_W$ be the input to \cref{alg:resolving-conflicts-in-layered} and let $\cA$ be its output. Then,
    \[
        \E{\sum_{C \in \cA} \gain(C)} \ge \frac{\eps^{13}}{2 e} \sum_{P \in \cP_W} \gain(P).
    \]

    Moreover, \cref{alg:resolving-conflicts-in-layered} can be implemented in $O(1)$ MPC rounds and $O(n^\delta + 1/\poly(\eps))$ memory per machine, for any constant $\delta > 0$.
\end{lemma}
\begin{proof}
We now upper-bound the probability that a sub-walk of $P \in \cP_W$ is part of the output.

We also have that $P$ contains $k = |\tau^A| + |\tau^B| + 1$ vertices. Since vertices of $\decompress(V(J), b)$ repeat in $\cL(J', W, \tau^A, \tau^B)$ across \emph{different} layers (but not within the same), $P$ intersects at most $2 k^2$ other alternating walks from $\cP_W$ ($k^2$ walks might have edge-intersection, and $k^2$ of them might have vertex intersection). There is a way to define layered graph (including the number of layers they have), so that $k^2 \le 1/\eps^8$ (we refer a reader to \cite{gamlath2019weighted} for details). From \cref{line:reject-in-resolve-within-layered} we have that $P' \in \cP_W$ is processed with probability $\eps^9/2$. Hence, each walk that $P$ intersects is reject with probability at least
\[
    \prob{\text{each $P'$ intersecting $P$ is rejected by \cref{line:reject-in-resolve-within-layered}}} \ge (1 - \eps^9/2)^{2 k^2} \ge (1 - \eps^9/2)^{2/\eps^8} \ge 1/e,
\]
where we assume $\eps < 1 / 2$.
Hence, we have
\begin{align*}
    & \prob{\text{$C$ is added to $\cA$ on \cref{line:add-C-to-cA-rsolve-within-layered}}} \\
    \ge & \prob{\text{$P$ is not rejected on \cref{line:reject-in-resolve-within-layered}}} \cdot \prob{\text{each $P'$ intersecting $P$ is rejected by \cref{line:reject-in-resolve-within-layered}}} \\
    \ge & \frac{\eps^9}{2 e}.
\end{align*}
In addition, due to \cref{line:select-largest-gain-resolve-within-layered}, we have $\gain(C) \ge \gain(P) / k \ge \eps^4 \gain(P)$. This further implies
\[
    \E{\sum_{C \in \cA} \gain(C)} = \sum_{\substack{P \in \cP_W \\ \text{$C$ selected on \cref{line:select-largest-gain-resolve-within-layered}}}} \prob{\text{$C$ is added to $\cA$ on \cref{line:add-C-to-cA-rsolve-within-layered}}} \cdot \gain(C) \ge \frac{\eps^{13}}{2 e} \sum_{P \in \cP_W} \gain(P),
\]
completing the approximation analysis
    
\paragraph{MPC implementation.}
Since each $P \in \cP_W$ has size $O(1/\poly(\eps))$, each $P$ is processed locally. Therefore, the first for-loop requires no extra communication.

To implement \cref{line:test-intersection-resolve-within-layered}, each $C \in \cA'$ generates a list of all its vertices in $\decompress(V(J), b)$. All these vertices, from all $C \in \cA'$, are sorted. In this sorted list, if a vertex $v^i$ has a neighbor $v^i$, then the alternating walk $C$ that $v^i$ is coming from has an intersection. Hence, $v^i$ communicated this information back to $C$. All $C \in \cA'$ that received that there is an intersection are dropped from further calculation.

Sorting and exchanging this information between $v^i$ and $C$ can be done in $O(1)$ rounds by using primitives described in \cite{Goodrich2011SortingSA}.
\end{proof}

\subsubsection{Conflicts between different layered graphs}
We now describe an algorithm that performs conflict resolution between alternating walks that are obtained by layered graphs corresponding to different weights.

\begin{algorithm}[h]
\TitleOfAlgo{\AlgResolveBetweenLayered}
    \SetAlgoLined
    \KwData{Let $\cW$ be a set of \emph{distinct} weights such that each $W \in \cW$ is of the form $(1+\eps^4)^i$ for some $i \in \bbN_{\ge 0}$. Let $\{\cA_{W_i}\}_{W_i \in \cW}$ be a collection of alternating walks in $J$, where $\cA_{W_i}$ is the output of \AlgResolveWithinLayered for $W = W_i$.}
    
    Let $t$ be the smaller integer such that $(1 + \eps^4)^i \ge 1/\eps^{20}$.\label{line:resolve-between-layered-define-t}
    
    Partition $\cW$ as follows: for $0 \le j < t$, let $\cW_j = \{(1+\eps^4)^{j + t \cdot h}\in \cW \ |\ h \in \bbN_{\ge 0}\}$.
    \label{line:partition-cW}
    
    \For{$j = 0 \ldots t - 1$}{
        Let $\cR_j = \emptyset$.
        
        \For{$W \in \cW_j$}{
            \For{$C \in \cA_{W}$}{
                Let $\decompress_{\cap C}(V(J), b)$ be the same as define on \cref{line:define-decompress-cap-C} of \cref{alg:resolving-conflicts-in-layered}.
                
                \If{$\decompress_{\cap C}(V(J), b) \cap \decompress_{\cap C'}(V(J), b) = \emptyset$ for every $C' \in \cA_{W'}$ such that $W' > W$ and $W' \in \cW_j$\label{line:resolve-between-layered-test-intersection}} {
                    Add $C$ to $\cR_j$.
                }
            }
        }
    }

    Let $\jstar$ be the index that maximizes $\sum_{C \in \cR_{j}} \gain(C)$ over all $j = 0 \ldots t - 1$.
        
    \Return{$\cR_{\jstar}$}
    
\caption{An algorithm that given a collection of walks in $J$, along with vertices in $\decompress(V(J), b)$ they correspond to, finds a subset of them which is non-intersecting with respect to $\decompress(V(J), b)$.
\label{alg:resolving-conflicts-between-layered}}
\end{algorithm}

\begin{lemma}
\label{lemma:resolve-between-layered}
    Let $\cI = \{\cA_{W_i}\}_{W_i \in \cW}$ be the input to \AlgResolveBetweenLayered. Then
    \[
        \sum_{C \in \cR_{\jstar}} \gain(C) \ge \poly(\eps) \sum_{C \in \cI} \gain(C).
    \]
    Moreover, \AlgResolveBetweenLayered can be implemented in $O(1)$ rounds with $O(n^{\delta} + \poly(1/\eps))$ memory per machine.
\end{lemma}
\begin{proof}
    Let $\cA_{W_i}$ be the input to \AlgResolveBetweenLayered (\cref{alg:resolving-conflicts-between-layered}).
    For each $C \in \cA_W$, by definition of weighted layered graphs it follows that $\gain(C) \ge \eps^{12} W$ and $\gain(C) \le 2 W$.
    
    Consider $C \in \cR_j$, where $\cR_j$ is obtained by \cref{alg:resolving-conflicts-between-layered}. We now want to upper-bound the sum of gains of augmentations not added to $\cR_j$ due to existence of $C$; these not-added augmentations did not pass the if condition on \cref{line:resolve-between-layered-test-intersection}.
    
    Let $f(i)$ be the sum of gain of augmentations ruled out due to the existence of an augmentation in $\cA_{(1 + \eps^4)^i}$. Trivially, we have $f(0) = 0$. Now, an augmentation $C \in \cA_W$ and $W \in \cW_j$ \emph{directly} rules out at most $|C|$ augmentations in each $\cA_{W'}$ such that $W' < W$ and $W' \in \cW_j$. However, each of those augmentations can recursively rule out some other ones. Also, recall that an augmentation in $\cA_W$ has gain at most $2 W$ and that we have $|C| \le 1/\eps^4$ (we refer a reader to the proof of \cref{lemma:resolve-within-layered} and definition of layered graphs for a reason why it holds). Therefore, we have
    \[
        f(i) \le \sum_{h \ge 1} \rb{\frac{1}{\eps^4} \cdot \rb{2 (1+\eps^4)^{i - t \cdot h} + f(i - t \cdot h)}} = \frac{1}{\eps^4} \rb{(1+\eps^4)^{i} \sum_{h \ge 1} 2 (1+\eps^4)^{- t \cdot h}  + \sum_{h \ge 1} f(i - t \cdot h)}.
    \]
    By our choice of $t$, we have $(1+\eps^4)^{- t \cdot h} \le \eps^{20 \cdot h}$. Using that, we can further upper-bound $f(i)$ as
    \[
        f(i) \le \frac{1}{\eps^4} \rb{(1+\eps^4)^{i} \sum_{h \ge 1} 2 \eps^{20 \cdot h}  + \sum_{h \ge 1} f(i - t \cdot h)} \le \frac{1}{\eps^4} \rb{(1+\eps^4)^{i} 4 \eps^{20}  + \sum_{h \ge 1} f(i - t \cdot h)}.
    \]
    Therefore, we have
    \begin{align*}
        f(i) & \le  \sum_{h \ge 0} 
        \frac{4}{\eps^{4 (h + 1)}} (1+\eps^4)^{i - t \cdot h} \eps^{20} \\
        & = 4 (1+\eps^4)^{i} \eps^{20} \sum_{h \ge 0} 
        \frac{(1+\eps^4)^{- t \cdot h}}{\eps^{4 (h + 1)}} \\
        & \le 4 (1+\eps^4)^{i} \eps^{20} \sum_{h \ge 0} 
        \frac{\eps^{20 \cdot h}}{\eps^{4 (h + 1)}} \\
        & = 4 (1+\eps^4)^{i} \eps^{20} \sum_{h \ge 0} 
        \eps^{20 h - 4h - 4}\\
        & = 4 (1+\eps^4)^{i} \eps^{16} \sum_{h \ge 0} 
        \eps^{16 h} \\
        & \le 8 (1 + \eps^4)^i \eps^{16},
    \end{align*}
    for sufficiently small $\eps$, e.g., $\eps < 1/2$. On the other hand, $\gain(C) \ge \eps^{12} (1 + \eps^4)^i$. Hence, for $\eps < 1/2$, adding an augmentation to $\cR_j$ rules out only $\eps \cdot \gain(C)$ gain from other augmentations.
    
    In addition, all the weights $\cW$ are divided into $t \in \poly(1/\eps)$ weight classes (see \cref{line:resolve-between-layered-define-t} for the definition of $t$), and the algorithm outputs the one carrying largest gain. This completes the approximation analysis.
    
    \paragraph{MPC implementation.}
    The only non-trivial step is to implement \cref{line:resolve-between-layered-test-intersection}. This is done similarly as in \cref{lemma:resolve-within-layered}. We now provide more details.
    
    Each $C \in \cA_W$ generates a list of all its vertices in $\decompress(V(J), b)$ along with $W$, $j$ such that $W \in \cW_j$, and $C$. (Recall that $|C| \in \poly(1/\eps)$.) In particular, one such tuple looks like $(v^i, j, W, C)$, where $v^i \in \decompress(V(J), b)$ is a copy of $v$.

    All these tuples are sorted. In this sorted list, if a tuple $(v^i, j, W, C)$ has as its neighbor a tuple $(v^i, j, W', C')$ with $W' > W$, then $C$ is not added to $\cR_j$.
\end{proof}

\subsection{Proof of \cref{theorem:main-weighted}}
Our proof strategy is to use \cref{theorem:gamlath} together with the algorithms we described to obtain the advertised claim.

\paragraph{Existence of specific augmenting walks.} We begin by recalling our discussion from \cref{sec:existence-of-short-augmentations} that enables us to claim that if there exists a collection of alternating paths that have a specific property and augment a weighted $1$-matching, then there exists a collection of alternating walks that have that same property and augment a weighted $b$-matching. This further enables us to carry over the properties (listed in \cite{gamlath2019weighted}) that relevant alternating paths have to augmenting walks in the context of $b$-matchings. Those properties are guaranteed by the design of weighted layered graphs.

\paragraph{Finding augmenting walks.}
However, in addition to existence of specific augmenting walks, we want to be able to find these walks algorithmically. In \cref{section:layering-for-b-matchings} we describe a graph layering that preserves many of relevant augmenting walks. We point out that the assignment of $M$ over $\decompress(V(G), b)$, i.e., the set called $\tM$, performed as Step~(I) is the same for all layered graphs, where layered graphs are defined with respect to different $W$, $\tau^A$ and $\tau^B$.

Our next goal is to find many augmentations in the corresponding layered graph $\cL$ that can be mapped back to walks in $G$. To be able to ensure a proper mapping, we impose additional property on layered graphs, Step~(III). Then, in \cref{sec:;ayered-alternating-to-augmenting-walks} we provide an algorithm that performs mapping and prove its correctness (see \cref{lemma:proof-algorithm-extract-alternations}).

To find many augmentations in $\cL$, we follow the steps of our approach described in \cref{sec:our-approach-1+eps-unweighted} and of that described in \cite{gamlath2019weighted}; the latter invokes $(1+\delta)$-approximate maximum matching on $\cL'$, where $\cL'$ is obtained by removing all matched edges in the first and last layer of $\cL$.
By construction, there exists a placement of unmatched edges between the vertices of $\cL'$ that leads to large matching, and hence to many augmentations.
We use the fact that among all possible placements of unmatched edges, where each placement conforms the layering constraints described in \cref{section:recall-layering-weighted,section:layering-for-b-matchings}, a maximum $b$-matching finds a placement that leads to the largest matching.
Therefore, we find a large $b'$-matching in $\cL$ where, the same as in our approach in \cref{sec:our-approach-1+eps-unweighted}, all copies of a vertex $v$ in $H_i$ are contracted to a single vertex $v_{H, i}$ (i.e., we invoke $\compress(H_i)$) and similarly all copies of $v$ in $T_i$ are contracted to a single vertex $v_{T, i}$ (i.e., we invoke $\compress(T_i)$). We let $b'_{v_{H, i}}$ be defined as the number of copies of $v$ in $H_i$, and $b'_{v_{T, i}}$ be defined as the number of copies of $v$ in $T_i$.
For each arc $\vec{e} = (u, v)$ (whose orientation has been fixed randomly) is added an edge $\{u_{H, i}, v_{T, i + 1} \}$ if $w(e)$ is ``close'' to $\tau^B_i$.

Let $M'_{\compress}$ be the $(1+\delta)$-maximum $b'$-matching found in the $\compress$-version of $\cL'$ as described above. We now need to distribute $M'_{\compress}$ among the vertices of $\cL'$. First, the edges in $M'_{\compress}$ that are matched in $\cL'$ we keep as they are -- their endpoints are already assigned. For the remainder of the set $M'_{\compress}$ we invoke the procedure described in \cref{lemma:assigning-M-between-decompress}. Let the resulting matching, i.e., $M'_{\compress}$ distributed among $\cL'$, be called $M'$.

Let $\tM_{\cL'}$ be the matching $\tM$ restricted to $\cL'$.
Finally, to find a large number of alternating paths in $\cL$, we take the same steps as \cite{gamlath2019weighted}, and find alternating paths among the edges $M' \cup \tM_{\cL'}$, where each such path has an endpoint in the first and an endpoint in the last layer of $\cL'$. Such algorithm and its analysis is provided in \cref{alg:AlgAlternating}. Observe that that algorithm implements $\AlgAlternating$ described as (B) of \cref{theorem:gamlath} where we let $\Mblue = \tM_{\cL'}$, $\Mred = M' \setminus \tM$, and the vertices in $H_1$ are marked as special.

Recall that these alternating paths are found in $\cL'$, which differs from $\cL$ in that that the edges between $H_1$ and $T_1$ and those between $H_{|\tau^A|}$ and $T_{|\tau^A|}$ are removed. So, for each alternating path found in $\cL'$ we also append the edges extending it in the first and the last layer of $\cL$, if any such edge exists. We need to do this so to have a valid input to \AlgResolveWithinLayered (see \cref{alg:resolving-conflicts-in-layered}). Step~(5) from \cref{section:recall-layering-weighted} assures that each alternating path found in $\cL'$ as described and extended to the maximal one in $\cL$ satisfies the input requirement of \cref{alg:resolving-conflicts-in-layered}.

\paragraph{Conflict resolution.} \AlgAlternating finds alternating paths in $\cL$ which are then used to extract alternating walks in $G$. The correctness of that method is analyzed in \cref{sec:;ayered-alternating-to-augmenting-walks}. However, although each of those alternating walks can independently by applied in $G$, applying all of them might not lead to a valid $b$-matching. To resolve that, we make a filtering of the augmenting walk that results in a collection of them yielding a valid $b$-matching. The conflict resolution is provided in \cref{sec:conflict-resolution}. It guarantees several properties:
\begin{itemize}
    \item No two walks in the final output intersect when they are mapped to $\decompress(V(G), b)$. For the walks within the same layered graph, this is guaranteed by \cref{line:test-intersection-resolve-within-layered} of \cref{alg:resolving-conflicts-in-layered}. For the vertex-disjointness of the walks between different layered graphs is guaranteed by \cref{line:resolve-between-layered-test-intersection} of \cref{alg:resolving-conflicts-between-layered}.
    The edge-disjointness is guaranteed by the fact that two weights in $\cW_j$ are at least $\eps^{20}$ factor apart (see \cref{line:partition-cW} of \cref{alg:resolving-conflicts-between-layered}). Now, by construction of layered graphs we have that a layered graph defined with respect to $W_1$ and one defined with respect to $W_2$ for $W_1 / W_2 \le \eps^{20}$ are edge-disjoint (see Steps~(4) and~(6) in \cref{section:recall-layering-weighted}).
    \item An augmenting walk $c$ which is not even cycle is such that each endpoint $v$ of $c$ either corresponds to a copy of $v$ which is a free vertex, or the edge in $c$ incident to $v$ is matched. Moreover, if both endpoints are free, they do not correspond to the same copy of a vertex. This is guaranteed by the subroutine \cref{alg:extracting-alternations} (see \cref{lemma:proof-algorithm-extract-alternations}) and also by the way we invoke that subroutine.
\end{itemize}
This now implies that applying the final collection of alternating walks that $\AlgResolveBetweenLayered$ and $\AlgResolveWithinLayered$ (which we call jointly as $\AlgResolve$) output yields a valid $b$-matching.

\paragraph{Running time analysis.}
In \cite{gamlath2019weighted}, it is set  $\delta = \eps^{28 + 900/\eps^2}$. In addition, we need to account for the probability of choosing the orientation ``right'' (performed in \cref{section:layering-for-b-matchings}, Step~(III)) so that a given augmenting walk appears in layered graphs. That probability is at least $2^{-|\tau^B|} \ge 2^{-32/\eps^2}$, where we use fact that $|\tau^B| \le 32/\eps^2$ in \cite{gamlath2019weighted}. Hence, in our work we let $\delta = \eps^{28 + 932/\eps^2}$.

Combining \cref{lemma:resolve-within-layered,lemma:resolve-between-layered}  we obtain $\rresolve = \poly(\eps)$.

Notice that our conflict resolution methods guarantee sufficiently large gain in expectation. However, using standard techniques and with $O(\log n)$ repetition in parallel, this can be turned into ``with high probability'' statement. We now overview these ideas.
\\
As long as the current matching is not $(1+\eps)$-approximate, the gain our method achieves can be shown to be in expectation $X \in \exp(1/\eps)$ fraction of the maximum matching. Hence, repeating the entire process $O(X)$ times where after each invocation we improve the existing matching, we get that in expectation the obtained matching is $\eps$ fraction smaller than the maximum one. By Markov's inequality, with probability at least $1/2$ this matching is not more than $2\eps$ fraction smaller than a maximum one. Therefore, if we repeat the entire process $O(\log n)$ times in parallel, then with high probability at least one of those matchings will be a $(1+2 \eps)$-approximate one. So, we invoke this procedure with $\eps' = \eps/2$ approximation requirement and the claim holds.

Combining these results, \cref{theorem:gamlath}, MPC and streaming implementation (\cref{sec:MPC-implementation-1+eps-unweighted,sec:streaming-implementation-1+eps-unweighted}), and using \cref{theorem:main-1+eps-unweighted} to find $(1+\delta)$-approximate unweighted $b$-matchings we obtain the desired round/pass and memory complexity.

\bibliographystyle{alpha}
\bibliography{references}

\appendix

\end{document}